\documentclass{article}
\usepackage{amsmath}
\usepackage{amssymb}
\usepackage{amsfonts}
\usepackage{a4wide}
\usepackage{amsthm}
\usepackage{tikz}
\usepackage{float}
\usepackage{subcaption}
\usepackage{hyperref}
\usepackage{enumerate}

\usetikzlibrary{arrows}

\newcounter{NN}
\newtheorem{notation}[NN]{Notation}

\newtheorem{remark}[NN]{Remark}
\newtheorem{proposition}[NN]{Proposition}

\newtheorem{conjecture}[NN]{Conjecture}
\newtheorem{lemma}[NN]{Lemma}

\newtheorem{example}[NN]{Example}

\def\N{\mathbb{N}}
\def\A{\mathbf{A}}
\def\B{\mathbf{B}}
\def\Z{\mathbf{Z}}

\def\x{{\mathbf x}}
\def\f{{\mathbf f}}
\def\y{{\mathbf y}}
\def\b{{\mathbf b}}

\begin{document}

\title{Hypergraphs and Lotka-Volterra systems with linear Darboux polynomials}
\author{
Peter H.~van der Kamp\\[2mm]
Department of Mathematical and Physical Sciences,\\
La Trobe University, Victoria 3086, Australia.\\
Email: P.vanderKamp@LaTrobe.edu.au\\[7mm]
}

\date{orcid.org/0000-0002-2963-3528}

\maketitle

\begin{abstract}
We associate parametric classes of $n$-component Lotka-Volterra systems which admit $k$ additional linear Darboux polynomials, with admissible loopless hypergraphs of order $n$ and size $k$. We study the equivalence relation on admissible hypergraphs induced by linear transformations of the associated LV-systems, for $n\leq 5$. We present a new 13-parameter 5-component superintegrable Lotka-Volterra system, i.e. one that is not equivalent to a so-called tree-system. We conjecture that tree-systems associated with nonisomorphic trees are not equivalent, which we verified for $n<9$.
\end{abstract}

\section{Introduction}
In \cite{QTMK,trees} a one-to-one correspondence was established between trees on $n$ vertices and homogeneous $n$-component parametric classes of Lotka-Volterra (LV) systems with $n-1$ linear Darboux polynomials (DPs). These classes of so-called tree-systems contain $3n-2$ parameters and were shown to be maximally superintegrable, cf. \cite{KMMQ} where tree-systems were shown to be measure-preserving with a rational measure that is also preserved by their Kahan discretisation.

For tree-systems, the linear DPs  depend on only two variables, and they correspond to the edges of the associated tree. If we allow the DPs to depend on more than $m=2$ variables, the corresponding edges become hyperedges (connecting $m$ vertices), and hence the associated graphs become hypergraphs. An $n$-component LV-system with $k$ linear DPs can be identified with a hypergraph on $n$-vertices with $k$ hyperedges. However, not all hypergraphs correspond to a class of LV-systems. We will call the ones that do correspond to a class of LV-systems {\em admissible hypergraphs}. The hypergraph of an LV-system $A$ is not necessarily isomorphic to the hypergraph of an LV-system $B$ obtained from $A$ by a linear transformation of the variables. This is due to the fact that although DPs are invariant under linear transformations, the number of variables they depend on is not. Thus, linear transformations induce an equivalence relation on the set of admissible hypergraphs with a fixed number of hyperedges, and it is this equivalence relation that we are interested in. 

The paper is organised as follows. In the next section, we start with some preliminaries on DPs, and how these give rise to integrals for LV-systems. We provide necessary and sufficient  conditions on the matrix $\A$ and the vector $\b$ such that the linear combination, with nonzero coefficients $\alpha_i$ and distinct $i_j\neq i_k$ when $i\neq k$,
\begin{equation} \label{DP}
P_{i_1,\ldots,i_m} = \alpha_1 x_{i_1} + \alpha_2 x_{i_2} + \cdots + \alpha_m x_{i_m},
\end{equation}
is a DP for the {\em Lotka-Volterra system}
\begin{equation} \label{LVsys}
\frac{dx_i}{dt} = x_i\left(b_i + \sum_{j=1}^{n}A_{i,j}x_j\right),\quad i=1,\dots,n.
\end{equation}
Here the $b_i$ are the components of the $n$-dimensional vector $\b$, and $A_{i,j}$ the entries of the $n\times n$ matrix $\A$. We provide the general solution to these conditions and show that any other solution is a special case. Our classes of LV-systems depend several parameters whose values are not specified, and statements hold for generic values of these parameters. For specialized values of the parameters, it may happen that integrals are no longer functionally independent (cf. \cite[Theorem 26]{trees}), or that one or more DPs of the form \eqref{DP} are parametric families of DPs, cf. Prop. \ref{p3}. We restrict our study of classes of LV-systems with several DPs, to the study of the intersection of general solutions. In section \ref{s3}, we introduce the reader to hypergraphs. We give some examples of admissible and nonadmissible hypergraphs, as well as equivalent and nonequivalent ones (with respect to linear transformation of their associated LV-systems). In section \ref{s4}, we provide all equivalent classes of admissible hypergraphs on $n$ vertices, for $2\leq n\leq 5$, and for each of them we provide a representative LV-system. We have found one 5-component maximally superintegrable LV-system that is not equivalent to a tree-system. We expect there exist higher component maximally superintegrable LV-systems that are not equivalent to tree-systems. We have verified, for $n<9$, that $n$-component LV-systems whose trees are not isomorphic are not equivalent. We conjecture the statement is true for all $n$. 

\section{Linear Darboux polynomials for LV-systems} \label{s2}
A polynomial $P(\x)$ is a {\em Darboux Polynomial} for a ODE $\frac{d\x}{dt}=\f(\x)$ (where $\f$ is a polynomial vector), if there exists another polynomial $C(\x)$, called the {\em cofactor}, such that $\frac{dP}{dt} = CP$. One observes that each $n$-component LV system, Eq. \eqref{LVsys}, has at least $n$ DPs, namely the variables $x_i, \, i=1,\dots,n$.

An {\em integral} is a DP with cofactor 0. As a product of powers of DPs is a DP whose cofactor is a linear combination of the cofactors of the DPs, constructing integrals becomes a linear algebra problem once one has sufficiently many DPs, cf. \cite{Gor}. Thus one can say that LV systems are on the verge of having an integral; one more linear DP suffices. An elementary method to determine $k$ integrals from $k$ additional linear DPs for a homogeneous LV-system (Eq. \eqref{LVsys} with $\b=\mathbf{0}$), given in \cite[Section 2]{QTMK}, is the following. Let $\B$ be the $k\times n$ matrix whose $i$th row contains the coefficients of the cofactor of the $i$th DP, $P_i$, i.e. we have, for $i=1,\ldots,k$,  $\dot{P}_i=P_i\sum_j B_{i,j}x_j$. If $\A$ is nonsingular, then $k$ integrals are given by, with $\Z=-\B\A^{-1}$,
\begin{equation} \label{Ik}
P_{i}\prod_{j=1}^n x_j^{Z_{i,j}},\quad i=1,\ldots,k.
\end{equation}
For the general classes of LV-systems we consider, and for generic values of their parameters, we have $\det(\A)\neq 0$.
\begin{notation}
Let $\N_n=\{1,2,\ldots,n\}$. For all sets of $m$ distinct integers $I=\{i_1,\ldots,i_m\}\subset \N_n$ we denote $I^{\rm c}=\N_n \setminus I$. A DP of the form \eqref{DP}, i.e. linear in $x_i$ for all $i\in I$, with $\alpha_i\neq 0$, will be variously called an $I$-DP, or an $m$-DP.
\end{notation}
In \cite{QTMK}, it was also shown that 
\begin{equation} \label{2dp}
P_{i,k} = \alpha x_i + \beta x_k,\quad \alpha\beta\neq0,
\end{equation}
is a DP of Eq. \eqref{LVsys} if and only if, for some constant $b$ and all $j\in \{i,k \}^c$, we have
\begin{equation} \label{cost}
A_{i,j} = A_{k,j},\quad
b_i = b_k\ =\ b,\quad
\alpha (A_{k,k}-A_{i,k}) = \beta (A_{k,i}-A_{i,i}),
\end{equation}
and $(A_{k,k}-A_{i,k})(A_{k,i}-A_{i,i})\neq 0$. We note that this statement is only true when $\det(A)\neq0$, otherwise we also have that $P_{i,k}$ is a DP for all $\alpha,\beta$ if $A_{k,k}-A_{i,k}=A_{k,i}-A_{i,i}=0$ (which is a special case).

The expression \eqref{2dp} is a 2-DP and is associated to an edge in a graph. In \cite{trees}, it was shown that if one chooses $n-1$ 2-DPs so that the associated graph is a tree, there exists a matrix $\A$, with $3n-2$ free parameters, for which all constraint \eqref{cost} are satisfied, and, that for generic values of these parameters, the corresponding $n-1$ integrals are functionally independent \cite{trees}. We will now consider $m$-DPs with $m>2$.

\begin{lemma} \label{lem}
The expression $P_I$ given by \eqref{DP}, with nonzero coefficients $\alpha_j\neq 0,\ j=1,\ldots,m$, is an $m$-DP for \eqref{LVsys} if and only if the following conditions are satisfied:
\begin{align}
&\exists b\ \forall i\in I\ b_{i} = b \tag{C1} \label{cond1}\\
&\forall i,k\in I\ \forall j \in I^{\rm c}\ A_{i,j} = A_{k,j} \tag{C2} \label{cond2}\\
&\forall j,h\in \N_m\ \alpha_h (A_{i_h,i_j}-A_{i_j,i_j}) = \alpha_j (A_{i_h,i_h}-A_{i_j,i_h}) \quad \tag{C3} \label{cond3}
\end{align}
\end{lemma}
\begin{proof}
The expression $P = P_I$, as given by Eq. \eqref{DP}, is a DP for LV-system \eqref{LVsys} iff there is an affine function
\[
C=\beta_0+\sum_{i=1}^n\beta_i x_i
\]
such that
\[
\dot{P}-CP=\sum_{i=1}^n d_ix_i +\sum_{1\leq i\leq j\leq n} c_{i,j} x_ix_j=0.
\]
Hence, $P$ is a  DP iff $d_i=0$ for all $i\in \N_n$ and $c_{i,j}=0$ for all $i,j \in \N_n$ such that $i\leq j$. We have
\[
d_i=
\begin{cases}
\alpha_i(nb_i-\beta_0)  & i\in I,\\
0                       & i\in I^{\rm c}.
\end{cases}
\]
Defining $b=\beta_0/n$, we have $b_i=b$ for all $i\in I$, \eqref{cond1}. For $i_k\in I$ and $j\in I^{\rm c}$ we have
$c_{i_k,j}=\alpha_k(A_{i_k,j}-\beta_j)$, the vanishing of which is equivalent to condition \eqref{cond2}, and it yields, for some choice of $h\in\N_m$, values of coefficients of the cofactor, $\beta_j=A_{i_h,j},\quad j\in I^{\rm c}$.
For $i_j\in I$, we find $c_{i_j,i_j}=\alpha_j(A_{i_j,i_j}-\beta_{i_j})$. So the values of the remaining coefficients of the cofactor are
\begin{equation} \label{bet}
\beta_j=A_{j,j},\quad j\in I.
\end{equation}
For $i,j\in I^{\rm c}$ the coefficients $c_{i,j}$ and $c_{i,i}$ are identically zero. For $i_k,i_j\in I$ and $i_j\neq i_k$, we have
\begin{equation} \label{cij}
c_{i_k,i_j}=(A_{i_k,i_j}-\beta_{i_j})\alpha_k+(A_{i_j,i_k}-\beta_{i_k})\alpha_j,
\end{equation}
which after substitution of \eqref{bet} yields \eqref{cond3}.
\end{proof}

We show, in Prop. \ref{p1}, how to obtain a general solution to the set of
equations \eqref{cond3}, under the constraint that $\alpha_j\neq 0,\ j=1,\ldots,m$.
We then show, in Prop. \ref{p2}, that any other general solution is a reparametrisation, and finally, in Prop \ref{p3}, that any other solution is a special case of the general one.

Given an index set $I$, a set of indices $\{i,j\}\subset I$ is called an edge, and a set of edges is called a graph on $I$. The graphs we consider are simple, i.e. $\{i,j\}$ is an edge implies $i\neq j$. The edges of a graph on $I=\{i_1,\ldots,i_m\}$ will be variously denoted by $e=\{x,y\}=\{i_{u},i_{v}\}\subset I$. The notation $e=\{x,y\}$ allows us to avoid having to use an index on an index, and the relation $x=i_{u}$ enables us to access the index $u$ of the element $x\in I=\{i_1,\ldots,i_m\}$. In terms of the complete graph on $I$, which has $m(m-1)/2$ edges and is denoted $K_I$, the set of equations \eqref{cond3} can be written as
\begin{equation} \label{ce}
\forall e\in K_I\ c(e)=0,
\end{equation}
where $c(e) := \alpha_{u}(A_{y,y}-A_{x,y}) - \alpha_{v} (A_{y,x}-A_{x,x})$. A tree $T$ on $I$ is a set of $m-1$ edges $e_j$, $j=1,\ldots, m-1$, such that $\cup_{j=1}^{m-1} e_j=I$. Given a tree $T$ on $I$, any two elements $x,y\in I$ are connected by the edges of a path in $T$, denoted $T_{x,y}\subset T$.
\begin{proposition}[General solution]
\label{p1} 
Let $T$ be a tree on $I$, such that for all $e=\{x,y\}\in T$ we have
\begin{equation} \label{neqs}
(A_{x,y}-A_{y,y})(A_{x,x}-A_{y,x})\neq 0.
\end{equation}
Then, for any $h\in \N_m$, a general solution to \eqref{cond3} is given by
\begin{equation} \label{sal}
\alpha_j=\prod_{e\in T_{i_h,i_j}}\frac{A_{x,y}-A_{y,y}}{A_{x,x}-A_{y,x}}\alpha_h,\quad \forall j\in \N_m,
\end{equation}
and, for all $e=(x,y) \in K_I\setminus T$, $x<y$,
\begin{equation} \label{sax}
A_{x,y}=A_{y,y}+\left(A_{x,x}-A_{y,x}\right)\prod_{e'\in T_{x,y}}\frac{A_{x',y'}-A_{y',y'}}{A_{x',x'}-A_{y',x'}}.
\end{equation}
\end{proposition}
\begin{proof}
The set of equations \eqref{ce} is the disjoint union of sets
\[
\{c(e)=0, e\in T\} \cup \{c(e)=0, e\in K_I\setminus T\}.
\]
The first set of $m-1$ equations determines the coefficients of $P_I$, up to an constant, and the second set provides $(m-2)(m-1)/2$ linear relations amongst the elements of $\A$.
Let $T_{i_h,i_j}=\cup_{j=1}^r \{e_j\}$, where $e_j=\{x_j,y_j\}=\{i_{u_j},i_{v_j}\}$ with $u_1=h$ and $v_r=j$. Then we have
\begin{align*}
\alpha_j&=\frac{A_{x_r,y_r}-A_{y_r,y_r}}{A_{x_r,x_r}-A_{y_r,x_r}}\alpha_{u_r}\\
&=\frac{A_{x_r,y_r}-A_{y_r,y_r}}{A_{x_r,x_r}-A_{y_r,x_r}}
\frac{A_{x_{r-1},y_{r-1}}-A_{y_{r-1},y_{r-1}}}{A_{x_{r-1},x_r{r-1}}-A_{y_{r-1},x_{r-1}}}\alpha_{u_{r-1}}\\
&\ \ \vdots\\
&=\prod_{k=1}^r \frac{A_{x_k,y_k}-A_{y_k,y_k}}{A_{x_k,x_k}-A_{y_k,x_k}}\alpha_{u_1}\\
&=\prod_{e\in T_{i_h,i_j}} \frac{A_{x,y}-A_{y,y}}{A_{x,x}-A_{y,x}} \alpha_{h}.
\end{align*}
and we note that for $j=h$ we have $T_{i_h,i_j}=\emptyset$. This establishes \eqref{sal}.
Now, for $e\in K_I\setminus T$ we have $(A_{x,x}-A_{y,x})\alpha_{v}=(A_{x,y}-A_{y,y})\alpha_{u}$ which we can solve for $A_{x,y}$ or $A_{y,x}$ (none of the other equations depend on either of these). We choose
\[
A_{x,y}=A_{y,y}+(A_{x,x}-A_{y,x})\frac{\alpha_{v}}{\alpha_{u}},\quad x<y.
\]
Substitution of \eqref{sal}, with $h=u$, yields \eqref{sax}. The solution is a general solution as we have not assumed any relations on the elements of $\A$, apart from the inequalities \eqref{neqs}.
\end{proof}

\begin{proposition}[Reparametrisation]\label{p2}
Any other general solution is a reparametrisation.
\end{proposition}
\begin{proof}
The solution obtained in Prop. \ref{p1} depends on $m^2-(m-1)(m-2)/2=m+2(m-1)+(m-1)(m-2)/2$ parameters, namely
\[
\{A_{x,x}: x\in I\} \cup \{A_{x,y}: \{x,y\}\in T\} \cup \{A_{y,x}: \{x,y\}\in K_I\setminus T,\ y>x\}.
\]
We rewrite the constraints on $\A$, cf. \eqref{sax}, as
\begin{equation} \label{pi}
\frac{A_{x,y}-A_{y,y}}{A_{x,x}-A_{y,x}}=\prod_{e'\in T_{x,y}}\frac{A_{x',y'}-A_{y',y'}}{A_{x',x'}-A_{y',x'}},\ \forall \{x,y\}\in K_I\setminus T.
\end{equation}
In fact, \eqref{pi} also holds for $\{x,y\}\in T$ (trivially). As the left-hand-side of \eqref{pi} does not depend on $T$, the right-hand-side does not depend on $T$, and neither does \eqref{sal}. The choice of a tree corresponds to a selection of the parameters.
\end{proof}

\begin{remark} \label{oc}
A nice choice of tree for the general solution is the bushy tree, $T=\cup_{i_j\neq i_h} \{\{i_j,i_h\}\}$. The expressions for the coefficients of $P_I$ in terms of $\alpha_h$ each involve only one ratio, and the constraints on $\A$ all have degree 3, i.e. $\forall j,k\in \N_m\setminus\{h\}$ $C_{h,j,k}=0$, where, cf. \cite[Eq. (1.5)]{Mai},
\begin{equation}\label{rc}
\begin{split}
C_{h,j,k}&=(A_{i_k,i_h}-A_{i_h,i_h})(A_{i_j,i_j}-A_{i_h,i_j})
(A_{i_j,i_k}-A_{i_k,i_k}) \\
&\ \ \ +(A_{i_j,i_h}-A_{i_h,i_h})
(A_{i_k,i_j}-A_{i_j,i_j})(A_{i_k,i_k}-A_{i_h,i_k}).
\end{split}
\end{equation}
\end{remark}

We will next look for special solutions, where not all the coefficients are uniquely determined up to a constant. If for some $e=\{x,y\}=\{i_u,i_v\}$ we would have $A_{x,y}-A_{y,y}\neq 0$, but $A_{x,x}-A_{y,x}= 0$, then $T$ cannot contain the edge $e$, as this would imply $\alpha_{u}=0$. For there to be no relation between the nonvanishing coefficients $\alpha_{u}$ and $\alpha_v$ we need
$A_{x,y}-A_{y,y}=A_{x,x}-A_{y,x}= 0$ but also,
\[
\prod_{e'\in T_{x,y}} A_{x',y'}-A_{y',y'}= \prod_{e'\in T_{x,y}} A_{x',x'}-A_{y',x'} = 0,
\]
for all possible paths $T_{x,y}$. The latter would be the case if $I=J\cup H$ where $x\in J$ and $y\in H$ and $\forall j\in J$ $\forall h\in H$ $A_{j,h}-A_{h,h}=A_{j,j}-A_{h,j}= 0$.

\begin{proposition}[Special solution] \label{p3}
Let $I=J\cup H$ and $\forall j\in J\ \forall h\in H$ $A_{j,h}-A_{h,h}=A_{j,j}-A_{h,j}= 0$. A special solution can be obtained as follows. Let $T$ be a tree on $J$ and $U$ be a tree on $H$ such that for all $e\in T\cup U$ we have $(A_{x,y}-A_{y,y})(A_{x,x}-A_{y,x})\neq 0$.
Then, with $i_h\in J$ and $i_k\in H$, a special solution to \eqref{cond3} is given by
\begin{align*}
\alpha_j&=\prod_{e\in T_{i_h,i_j}}\frac{A_{x,y}-A_{y,y}}{A_{x,x}-A_{y,x}}\alpha_h,\quad \forall i_j\in J,\\
\alpha_j&=\prod_{e\in U_{i_k,i_j}}\frac{A_{x,y}-A_{y,y}}{A_{x,x}-A_{y,x}}\alpha_k,\quad \forall i_j\in H,
\end{align*}
and, with $x<y$,
\begin{align*}
A_{x,y}&=A_{y,y}+\left(A_{x,x}-A_{y,x}\right)\prod_{e'\in T_{x,y}}\frac{A_{x',y'}-A_{y',y'}}{A_{x',x'}-A_{y',x'}},\ \forall e \in K_J\setminus T,\\
A_{x,y}&=A_{y,y}+\left(A_{x,x}-A_{y,x}\right)\prod_{e'\in U_{x,y}}\frac{A_{x',y'}-A_{y',y'}}{A_{x',x'}-A_{y',x'}},\ \forall e \in K_H\setminus U.
\end{align*}
The corresponding LV-system (for which \eqref{cond1} and \eqref{cond2} are satisfied) has a $J$-DP and a $H$-DP, the sum of which is an $I$-DP.
\end{proposition}
\begin{proof} The conditions \eqref{cond3} are satisfied, so $P_I$ is an $m$-DP whose coefficients depend on 2 free parameters, $\alpha_h$ and $\alpha_k$. The expression $P_J=\sum_{i_j\in J}\alpha_jx_{i_j}$ is a $J$-DP and $P_H=\sum_{i_j\in H}\alpha_jx_{i_j}$ is a $H$-DP, as the conditions in Lemma \ref{lem} (with $I$ replaced by $J$ resp. $H$) are satisfied. Note that for $P_J$ the condition \eqref{cond2} reads
\[
\forall i,k\in J\ \forall j \in J^{\rm c}\ A_{i,j} = A_{k,j},
\]
and we have $J^{\rm c}=I^c\cup H$. The $I^c$ part is covered by condition \eqref{cond2} for $P_I$ and the $H$ part is covered by $\forall j\in J\ \forall h\in H$ $A_{j,h}-A_{h,h}=0$. Similarly for $P_H$ the $J$-part of \eqref{cond2} holds due to $\forall j\in J\ \forall h\in H$ $A_{j,j}-A_{h,j}=0$.
\end{proof}
If $J$ has $f$ elements and $H$ has $g$ elements, the special solution described in Prop. \ref{p3} is a $(f(f + 3)+g(g+3))/2-2$ parameter subfamily of the general solution. It should be clear that we can further specialise the parameters in special solutions to obtain more than 2 free parameters in the $m$-DP.\\

In the next example, we consider the case $m=3$ and $n=4$.
\begin{example}
We describe when
$P=\alpha_1x_1+\alpha_2x_2+\alpha_3x_3$ is a DP of the system \eqref{LVsys}, with
\begin{equation} \label{NS}
\b=\begin{pmatrix} b_1 \\ b_2 \\ b_3 \\ b_4 \end{pmatrix},\quad \A=\begin{pmatrix}
a_{1} & a_{2} & a_{3} & a_{4} 
\\
 a_{5} & a_{6} & a_{7} & a_{8} 
\\
 a_{9} & a_{10} & a_{11} & a_{12} 
\\
 a_{13} & a_{14} & a_{15} & a_{16} 
\end{pmatrix}.
\end{equation}
The condition \eqref{cond1} tells us that we need $b_1=b_2=b_3$. The condition \eqref{cond2} tells us that we need $a_{4}=a_{8}=a_{12}$. And \eqref{cond3} comprises
\begin{align}
\alpha_{1} \left(a_{6}-a_{2}\right)&= 
\alpha_{2} \left(a_{5}-a_{1}\right)
,\label{c1}\\
\alpha_{1} \left(a_{11}-a_{3}\right)&= 
\alpha_{3} \left(a_{9}-a_{1}\right)
,\label{c2}\\
\alpha_{2} \left(a_{11}-a_{7}\right)&= 
\alpha_{3} \left(a_{10}-a_{6}\right)
.\label{c3}
\end{align}
We take
\begin{equation} \label{cds}
(a_{1}-a_{5})(a_{2}-a_{6})(a_{6}-a_{10})(a_{7} -a_{11})\neq 0,
\end{equation}
which corresponds to $T=\{\{1,2\},\{1,3\}\}$. Solving \eqref{c1},\eqref{c3} for $\alpha_1,\alpha_3$ gives
\[
\alpha_{1} = 
\frac{a_{1}-a_{5}}{a_2-a_{6}}\alpha_{2},\quad \alpha_{3}
 = \frac{a_{7}-a_{11}}{a_{6}-a_{10}}\alpha_{2} 
\]
and the constraint on $\A$, \eqref{c2}, becomes
\[
\left(a_{1}-a_{5}\right) \left(a_3-a_{11}\right) \left(a_{6}-a_{10}\right)
 = 
\left(a_{7}-a_{11}\right) \left(a_1-a_{9}\right) \left(a_2-a_{6}\right),
\]
i.e. $C_{1,2,3}=0$ with $C_{1,2,3}$ given by \eqref{rc}. Solving for $
a_3=a_{11}+\frac{\left(a_{7}-a_{11}\right) \left(-a_{9}+a_{1}\right) \left(-a_{6}+a_{2}\right)}{\left(a_{1}-a_{5}\right) \left(a_{6}-a_{10}\right)}$
yields a 15 parameter family of LV-systems, with
\begin{equation} \label{JOS}
\b=\begin{pmatrix} b_1 \\ b_1 \\ b_1 \\ b_4 \end{pmatrix},\quad
\A=\begin{pmatrix}
a_{1} & a_{2} & a_{11}+\frac{\left(a_{7}-a_{11}\right) \left(-a_{9}+a_{1}\right) \left(-a_{6}+a_{2}\right)}{\left(a_{1}-a_{5}\right) \left(a_{6}-a_{10}\right)} & a_{4} 
\\
 a_{5} & a_{6} & a_{7} & a_{4} 
\\
 a_{9} & a_{10} & a_{11} & a_{4} 
\\
 a_{13} & a_{14} & a_{15} & a_{16} 
\end{pmatrix}
\end{equation}
which admits the 3-DP,
\[
\left(a_{6}-a_{10}\right) \left(a_{1}-a_{5}\right)x_{1}
+\left(a_2-a_{6}\right) \left(a_{6}-a_{10}\right)x_{2}
+\left(a_2-a_{6}\right) \left(a_{7}-a_{11}\right) x_{3},
\]
where we have set $\alpha_{2} = \left(a_2-a_{6}\right) \left(a_{6}-a_{10}\right)$. We obtain a special solution if, choosing $J=\{1\}$, $H=\{2,3\}$, we impose
\begin{equation} \label{so}
a_{5} = a_{1},\
a_{2} = a_{6},\
a_{9} = a_{1},\
a_{3} = a_{11}.
\end{equation}
Now the conditions \eqref{c1} and \eqref{c2} are satisfied, and \eqref{c3} yields
$
\alpha_3=\frac{a_{7}-a_{11}}{a_{6}-a_{10}}\alpha_{2}$.
This gives a 12 parameter subfamily of LV-systems, with
\begin{equation} \label{SOS}
\b=\begin{pmatrix} b_1 \\ b_1 \\ b_1 \\ b_4 \end{pmatrix},\quad
\A=\begin{pmatrix} 
a_{1} & a_{6} & a_{11} & a_{4} 
\\
 a_{1} & a_{6} & a_{7} & a_{4} 
\\
 a_{1} & a_{10} & a_{11} & a_{4} 
\\
 a_{13} & a_{14} & a_{15} & a_{16} 
\end{pmatrix}
\end{equation}
which admits the 3-DP (we set $\alpha_1=c_1$, $\alpha_2=(a_6-a_{10})c_2$),
\[
c_1x_1+c_2\left((a_6-a_{10})x_2+(a_{7}-a_{11})x_3\right),
\]
where $c_1,c_2$ are arbitrary constants.
The more special case where besides \eqref{so} we also have $a_{10}=a_6$ and $a_{7}=a_{11}$ gives rise to a class of LV-systems, which admits the 3-DP $c_1x_1+c_2x_2+c_3x_3$, where $c_1,c_2,c_3$ are free parameters. 
\end{example}

In this section, we have studied LV-systems with one $I$-DP. In the remaining of this paper, we will study general classes of LV-systems with several $I_i$-DPs, with $I_i\subset \N_n$, $i=1\ldots k\in \N$, by which we mean classes that are the intersection of the general solutions for each $i$. We will not study special subclasses. (In the next section, we do provide a generic example of a general class and, in example 8, we mention one specific special subclass.)

\section{Admissible hypergraphs, equivalence} \label{s3}
By a hypergraph on $n$ vertices (undirected) we mean a subset of the powerset of $N=\{1,2,\ldots,n\}$. The number of vertices is the order and the number of (hyper)edges is the size of the hypergraph. An example of a hypergraph of order $3$ and size $6$ is
$
\{\{1\},\{2\},\{3\},\{1,3\},\{1,2,3\}\}.
$
This would be the hypergraph associated to the $3$-component LV-system with two additional Darboux polynomials
\begin{equation} \label{2dps}
x_1+2x_3,\quad 3x_1+4x_2+6x_3.
\end{equation}
Two hypergraphs are isomorphic if they can be obtained from each other by a permutation of $N$, e.g. the above hypergraph is isomorphic to $\{\{1\},\{2\},\{3\},\{2,3\},\{1,2,3\}\}$.
Because we can relabel the variables of the LV-system, we are only interested in nonisomorphic hypergraphs. Moreover, because we restrict ourselves to the study of LV-systems, and their hypergraphs always contain the full set of singletons $\{\{1\},\ldots,\{n\}\}$, we may as well omit them. In the sequel we assume our hypergraphs do not contain loops (and also no multiple edges). Thus, to the LV-system with DPs \eqref{2dps} we associate the loopless hypergraph of size 2,
$
\{\{1,3\},\{1,2,3\}\},
$
reflecting only the additional DPs. The number of nonisomorphic hypergraphs of order $n$ grows quite rapidly \cite[A000612]{OEIS}
\[
1, 2, 6, 40, 1992, 18666624, 12813206169137152, 33758171486592987164087845043830784, \ldots.
\]
The triangular array counting the number of hypergraphs of order $n$ and size $k$ is given by
\cite[A371830]{OEIS}. Excluding  singletons, the number of loopless hypergraphs of order $n$ is \cite[A317794]{OEIS}
\[
1, 1, 2, 8, 180, 612032, 200253854316544, 263735716028826427534807159537664, \ldots.
\]
We have generated all loopless hypergraphs of order $n$ and size $k$ for $n\leq 5$ and, when $n=5$, $k\leq 7$. Their number is given in Table \ref{nhg}, cf. \cite{HS}.

\begin{table}[h]
\begin{center}
\begin{tabular}{c|cccccccccccc}
$n \setminus k$ & 0 & 1 & 2 & 3 & 4 & 5 & 6 & 7 & 8 & 9 & 10 & 11\\
\hline
2 & 1 & 1 & & & & & & & & & & \\
3 & 1 & 2 & 2 & 2 & 1 & & & & & & &\\
4 & 1 & 3 & 7 & 16 & 28 & 35 & 35 & 28 & 16 & 7 & 3 & 1\\
5 &1 & 4 & 15 & 62 & 243 & 841 & 2544 & 6672 & & & \\
\end{tabular}
\vspace{3mm}

\parbox{12.5cm}{\caption{\label{nhg} The number of loopless hypergraphs of order $n$ and size $k$.}}
\end{center}
\end{table}
\vspace{-8pt} 

For $n=3$, the set of nonisomorphic hypergraphs is given in Fig. \ref{nih3}.
\begin{figure}[h]
\centering
\begin{subfigure}{0.15\textwidth}
\begin{tikzpicture}[scale=1]
\tikzstyle{nod}= [circle, inner sep=0pt, fill=white, minimum size=2pt, draw]		
\node[nod] (a) at (1,0) {};
\node[nod] (b) at (2,0) {};
\node[nod] (c) at (3/2,1) {};	
\node at (3/2,-1) {$k=0$};
\end{tikzpicture}
\end{subfigure}
\begin{subfigure}{0.15\textwidth}
\begin{tikzpicture}[scale=1]
\tikzstyle{nod}= [circle, inner sep=0pt, fill=white, minimum size=2pt, draw]		
\node[nod] (a) at (1,0) {};
\node[nod] (b) at (2,0) {};
\node[nod] (c) at (3/2,1) {};
\draw[line width=1pt] (a) --  (b);
\node[nod] (aa) at (1,2) {};
\node[nod] (bb) at (2,2) {};
\node[nod] (cc) at (3/2,3) {};
\draw[line width=1pt] (1.14,2.1) --  (1.86,2.1) -- (1.5,2.83) -- (1.14,2.1);
\node at (3/2,-1) {$k=1$};	
\end{tikzpicture}
\end{subfigure}
\begin{subfigure}{0.15\textwidth}
\begin{tikzpicture}[scale=1]
\tikzstyle{nod}= [circle, inner sep=0pt, fill=white, minimum size=2pt, draw]		
\node[nod] (a) at (1,0) {};
\node[nod] (b) at (2,0) {};
\node[nod] (c) at (3/2,1) {};
\draw[line width=1pt] (c) -- (a) --  (b);
\node[nod] (aa) at (1,2) {};
\node[nod] (bb) at (2,2) {};
\node[nod] (cc) at (3/2,3) {};
\draw[line width=1pt] (1.14,2.1) --  (1.86,2.1) -- (1.5,2.83) -- (1.14,2.1);
\draw[line width=1pt] (aa) --  (bb);
\node at (3/2,-1) {$k=2$};	
\end{tikzpicture}
\end{subfigure}
\begin{subfigure}{0.15\textwidth}
\begin{tikzpicture}[scale=1]
\tikzstyle{nod}= [circle, inner sep=0pt, fill=white, minimum size=2pt, draw]		
\node[nod] (a) at (1,0) {};
\node[nod] (b) at (2,0) {};
\node[nod] (c) at (3/2,1) {};
\draw[line width=1pt] (c) -- (a) --  (b) -- (c);
\node[nod] (aa) at (1,2) {};
\node[nod] (bb) at (2,2) {};
\node[nod] (cc) at (3/2,3) {};
\draw[line width=1pt] (1.14,2.1) --  (1.86,2.1) -- (1.5,2.83) -- (1.14,2.1);
\draw[line width=1pt] (cc) -- (aa) --  (bb);
\node at (3/2,-1) {$k=3$};	
\end{tikzpicture}
\end{subfigure}
\begin{subfigure}{0.15\textwidth}
\begin{tikzpicture}[scale=1]
\tikzstyle{nod}= [circle, inner sep=0pt, fill=white, minimum size=2pt, draw]		
\node[nod] (a) at (1,0) {};
\node[nod] (b) at (2,0) {};
\node[nod] (c) at (3/2,1) {};
\draw[line width=1pt] (c) -- (a) --  (b) -- (c);
\draw[line width=1pt] (1.14,0.1) --  (1.86,0.1) -- (1.5,0.83) -- (1.14,0.1);
\node at (3/2,-1) {$k=4$};	
\end{tikzpicture}
\end{subfigure}
\caption{\label{nih3} Nonisomorphic hypergraphs on 3 vertices. Edges represent 2-DPs, and triangles represent 3-DPs (if the hypergraph is admissible).}
\end{figure}
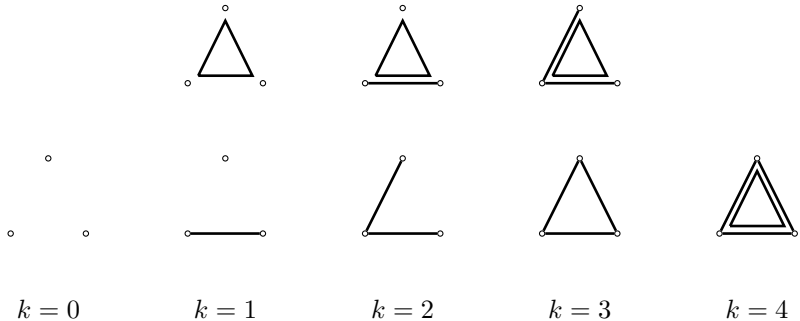

As we now know, cf. \cite{QTMK,trees}, for every tree there is a class of LV-systems, whose DPs/independent integrals correspond to the edges of the tree. The question arises whether for every hypergraph there is a corresponding general class of LV-systems. This is not the case, as is illustrated by the following examples. Hypergraphs with a corresponding general class of LV-systems will be called {\em LV-admissible}, or {\em admissible} for short. For admissible hypergraphs, we would be interested to know the number of free parameters. For now, and in Section 4, we will consider homogeneous LV-systems. However, the results hold more generally, at least in the nonhomogeneous case where $\b=b$ is constant, cf. \cite[Theorem 3.1]{HDP}. And we will get back to the nonhomogeneous case in Section 5.

\begin{example} \label{h34}
Consider the hypergraph on three vertices with $k=4$ hyperedges, see Fig. \ref{nih3}. Necessary and sufficient conditions on the entries of
\[
\A=\begin{pmatrix}
a_{1} & a_{2} & a_{3} 
\\
 a_{4} & a_{5} & a_{6} 
\\
 a_{7} & a_{8} & a_{9} 
\end{pmatrix}
\]
for the LV-system to admit the 3-DP $\alpha_1x_1+\alpha_2x_2+\alpha_3x_3$ are
\begin{equation} \label{eqsc}
\alpha_{1} \left(a_{5}-a_{2}\right) = 
\alpha_{2} \left(a_{4}-a_{1}\right),\quad
\alpha_{1} \left(a_{9}-a_{3}\right) = 
\alpha_{3} \left(a_{7}-a_{1}\right),\quad
\alpha_{2} \left(a_{9}-a_{6}\right) = 
\alpha_{3} \left(a_{8}-a_{5}\right).
\end{equation}
To obtain a general solution, we choose the tree $\{\{1,2\},\{1,3\}\}$, so that 
\begin{equation} \label{nn}
\left(a_{5}-a_{2}\right) \left(a_{1}-a_{4}\right) \left(a_{9}-a_{3}\right) \left(a_{1}-a_{7}\right)\neq 0,
\end{equation}
and we can solve the first two equations of \eqref{eqsc} for $
\alpha_{2} = \frac{\alpha_{1} \left(a_{2}-a_{5}\right)}{a_{1}-a_{4}}$ and $\alpha_{3} = \frac{\alpha_{1} \left(a_{3}-a_{9}\right)}{a_{1}-a_{7}}$.
Choosing $\alpha_1=(a_{1}-a_{4})(a_{1}-a_{7})$ the 3-DP is given by
\[
\left(a_{1}-a_{4}\right) \left(a_{1}-a_{7}\right) x_{1}+\left(a_{1}-a_{7}\right) \left(a_{2}-a_{5}\right) x_{2}+\left(a_{1}-a_{4}\right) \left(a_{3}-a_{9}\right) x_{3},
\]
and none of the coefficients vanish due to \eqref{nn}. The last equation of \eqref{eqsc} becomes 
\begin{equation} \label{C123}
\left(a_{2}-a_{5}\right) \left(a_{1}-a_{7}\right) \left(a_{6}-a_{9}\right)=\left(a_{3}-a_{9}\right) \left(a_{1}-a_{4}\right) \left(a_{5}-a_{8}\right),
\end{equation}
which is $C_{1,2,3}=0$, cf. \eqref{rc}.
For the class of LV-systems to admit three 2-DPs, necessary requirements are
\begin{equation} \label{32dps}
a_{8} = a_{2}, a_{7} = a_{4}, a_{6} = a_{3}.
\end{equation}
Substituting \eqref{32dps} into \eqref{C123} yields
\begin{equation} \label{nz}
2 \left(a_{2}-a_{5}\right) \left(a_{1}-a_{4}\right) \left(a_{3}-a_{9}\right)=0.
\end{equation}
which violates \eqref{nn}. Hence, since there is no general class of LV-systems which admits the hypergraph on three vertices with four hyperedges, this hypergraph is not admissible.

We need to mention here that there exists a special class of LV-systems which admits this hypergraph. Consider the homogeneous class of LV-systems with singular matrix
\[
\A=\begin{pmatrix}
a_1 & a_2 & a_3 \\
a_1 & a_2 & a_3 \\
a_1 & a_2 & a_3 \\
\end{pmatrix}.
\]
It admits the $k=4$ hypergraph since any linear combination $\alpha_1x_1+\alpha_2x_2+\alpha_3x_3$ is a DP. Apart from being a special case of the class with 1 3-DP, it is also a special case of the class with matrix
\[
\A=\begin{pmatrix}
a_1 & a_2 & a_3 \\
a_4 & a_5 & a_3 \\
a_4 & a_2 & a_9 \\
\end{pmatrix},
\]
which has $k=3$ 2-DPs,
$\left(a_{1}-a_{4}\right) x_{1}-\left(a_{5}-a_{2}\right) x_{2}$,
$\left(a_{1}-a_{4}\right) x_{1}-\left(a_{9}-a_{3}\right) x_{3}$ and $\left(a_{5}-a_{2}\right) x_{2}-\left(a_{9}-a_{3}\right) x_{3}$.
If one of the factors of \eqref{nz} vanishes, then two of the 2-DPs become 1-DPs. If two of the factors vanish, then there is also a 2-DP with arbitrary coefficients.
\end{example}
In the next two examples, the conditions can be satisfied but there are consequences.

\begin{example} \label{exam1}
Consider the hypergraph on four vertices with two hyperedges as in Fig. \ref{h423}(a).
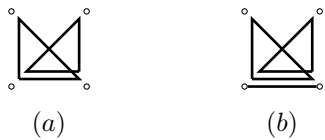
\begin{figure}[H]
\centering
\begin{subfigure}{0.2\textwidth}
\begin{tikzpicture}[scale=1]
\tikzstyle{nod}= [circle, inner sep=0pt, fill=white, minimum size=2pt, draw]		
\node[nod] (a) at (0,0) {};
\node[nod] (b) at (1,0) {};
\node[nod] (c) at (0,1) {};
\node[nod] (d) at (1,1) {};
\draw[line width=1pt] (0.1,0.1) --  (.1,.9) -- (.9,.1) -- (.1,.1);
\draw[line width=1pt] (.9,.2) -- (.2,.2) --  (.9,.9) -- (.9,.2);
\node at (1/2,-1/2) {$(a)$};	
\end{tikzpicture}
\end{subfigure}
\begin{subfigure}{0.2\textwidth}
\begin{tikzpicture}[scale=1]
\tikzstyle{nod}= [circle, inner sep=0pt, fill=white, minimum size=2pt, draw]		
\node[nod] (a) at (0,0) {};
\node[nod] (b) at (1,0) {};
\node[nod] (c) at (0,1) {};
\node[nod] (d) at (1,1) {};
\draw[line width=1pt] (0.1,0.1) --  (.1,.9) -- (.9,.1) -- (.1,.1);
\draw[line width=1pt] (.9,.2) -- (.2,.2) --  (.9,.9) -- (.9,.2);
\draw[line width=1pt] (a) -- (b);
\node at (1/2,-1/2) {$(b)$};	
\end{tikzpicture}
\end{subfigure}
\caption{\label{h423} Hypergraphs on 4 vertices with 2 resp. 3 hyperedges.}
\end{figure}
The LV-system with interaction matrix
\[
\A=\begin{pmatrix}
a_{1} & a_{2} & a_{3} & a_{4} 
\\
 a_{5} & a_{6} & a_{7} & a_{8} 
\\
 a_{9} & a_{10} & a_{11} & a_{12} 
\\
 a_{13} & a_{14} & a_{15} & a_{16} 
\end{pmatrix}
\]
has a $\{1,2,3\}$-DP if $C_{1,2,3}=0$ and $a_4=a_8=a_{12}$. It has a $\{1,2,4\}$-DP if
$C_{1,2,4}=0$ and $a_3=a_7=a_{15}$. So, as a consequence, we have $a_4=a_8$ and $a_3=a_7$,
which implies that the LV-system also admits a $\{1,2\}$-DP. Hence, the hypergraph in Fig. \ref{h423}(a) is not admissible. On the other hand, the hypergraph in Fig. \ref{h423}(b) is admissible (with
10 free parameters).
\end{example}
\begin{example} \label{exam2} Any admissible hypergraph with edges of degree 2 that form a loop contains the edges of degree 2 between all pairs of vertices of the loop, cf. \cite[Proposition 4.1]{KMMQ}.
\end{example}
Another issue we like to address is equivalence. LV-systems that are related to each other by a linear transformation, may correspond to different hypergraphs, which are not isomorphic. We will call such hypergraphs {\em LV-equivalent}, or {\em equivalent} for short.
\begin{example}
Consider the LV-system with interaction matrix
\[
\A=\begin{pmatrix}
1 & 1 & 1 
\\
 4 & 5 & 9 
\\
 2 & 1 & 3 
\end{pmatrix}.
\]
It admits the two DPs \eqref{2dps}, and these give rise to the following integrals
\[
\frac{\left(x_{1}+2 x_{3}\right)^3 x_{1}^{3}}{x_{2} x_{3}},\quad
\frac{\left(3 x_{1}+4 x_{2}+6 x_{3}\right) x_{3}^{3}}{x_{1}^{3} x_{2}}.
\]
By introducing the DPs
\[
y_{1} = x_{1},\quad
y_{2} = x_{1}+2 x_{3},\quad
y_{3} = 3 x_{1}+4 x_{2}+6 x_{3},
\]
as new variables, the equations for $\y$ take the form of an LV-system with interaction matrix
\[
\frac 14 \begin{pmatrix}
2 & -1 & 1 
\\
 -2 & 3 & 1 
\\
 -2 & -9 & 5 
\end{pmatrix}
\]
and, the $\y$-system admits the additional DPs $2x_3=y_2-y_1$ and $4x_2=y_3-3y_2$, i.e. it is a tree-system. Indeed, the two nonisomorphic hypergraphs in Fig. \ref{nih3} with $k=2$ are LV-equivalent (see section \ref{ss3}).
\end{example}

\section{Classification of $n$-component homogeneous LV-systems with linear DPs, for $n\leq5$} \label{s4}
For each $n\leq 5$, we consider the set of nonisomorphic hypergraphs (with $k\leq 7$ hyperedges if $n=5$) and determine which ones are admissible. Then, for each associated LV-system we consider the finite set of linear transformations where the new variables are Darboux polynomials, and hence yield LV-systems.
These systems may be associated to hypergraphs that are nonisomorphic but LV-equivalent. For $n<5$ we present the equivalence classes of admissible hypergraphs, for $n=5$ we provide nonequivalent ones.  
\subsection{2-component LV systems}
Not considering the empty graph, there is only one graph on two vertices: the tree \begin{tikzpicture}[scale=1]
\tikzstyle{nod}= [circle, inner sep=0pt, fill=white, minimum size=2pt, draw]		
\node[nod] (a) at (0,0) {};
\node[nod] (b) at (1,0) {};
\draw[line width=1pt] (a) -- (b);	
\end{tikzpicture}. The associated LV-system
\begin{equation} \label{xn2}
\begin{split}
\dot{x_{{1}}}&=x_{{1}} \left( a_{{1}}x_{{1}}+b_{{1}}x_{{2}} \right)\\
\dot{x_{{2}}}&=x_{{2}} \left( c_{{1}}x_{{1}}+a_{{2}}x_{{2}} \right).
\end{split}
\end{equation}
has 4 parameters, and admits the additional Darboux polynomial $P=(c_1 - a_1)x_1 + (a_2 - b_1)x_2$ and corresponding integral
$
x_{1}^{a_{4} \left(a_{3}-a_{1}\right)} x_{2}^{a_{1} \left(a_{2}-a_{4}\right)}P^{a_{1} a_{4}-a_{2} a_{3}} $.

\subsection{3-component LV-systems} \label{ss3}
There are 7 nonisomorphic hypergraphs of order 3 and positive size, cf. Fig. \ref{nih3}.
\begin{enumerate}[1.]
\item The two LV-systems related to the hypergraphs of size 1 are not LV-equivalent. Each of them has 8 free parameters. Their LV-matrices are respectively
\[
\begin{pmatrix}
a_{1} & a_{2} & a_{5} 
\\
a_{3} & a_{4} & a_{5} 
\\
a_{6} & a_{7} & a_{8} 
\end{pmatrix},\qquad
\begin{pmatrix}
a_{1} & a_4+\frac{\left(a_{2}-a_{8}\right) \left(a_{1}-a_{3}\right) \left(a_{4}-a_{7}\right)}{\left(a_{5}-a_{8}\right) \left(a_{1}-a_{6}\right)}  & a_{2} 
\\
 a_{3} & a_{4} & a_{5} 
\\
 a_{6} & a_{7} & a_{8} 
\end{pmatrix}.
\]
\item The bottom hypergraph of size 2 corresponds to the tree-system, with $3n-2=7$ parameters (here we have taken the parametrisation for tree-systems, as in \cite{trees}),
\begin{equation} \label{xn3}
\begin{split}
\dot{x_{{1}}}&=x_{{1}} \left( a_{{1}}x_{{1}}+b_{{1}}x_{{2}}+b_{{2}}x_{{3}}
 \right)\\
\dot{x_{{2}}}&=x_{{2}}\left(
c_{{1}}x_{{1}}+a_{{2}}x_{{2}}+b_{{2}}x_{{3}}
 \right)\\
\dot{x_{{3}}}&=x_{{3}} \left(
c_{{1}}x_{{1}}+c_{{2}}x_{{2}}+a_{{3}}x_{{3}}
\right).
\end{split}
\end{equation}
The edges are $e_1=(1,2)$, $e_2=(2,3)$, related to the additional DPs
\[
P_1=\left(c_{1}-a_{1}\right) x_{1}+\left(a_{2}-b_{1}\right) x_{2},\quad
P_2=\left(c_{2}-a_{2}\right) x_{2}+\left(a_{3}-b_{2}\right) x_{3}.
\]
Performing the linear transformation $\x\mapsto \y=(x_1,P_1,P_2)$ we obtain the LV system
\begin{subequations}\label{yn3}
\begin{equation}
\frac{dy_i}{dt} = y_i\left(\sum_{j=1}^{3}A_{ij}y_j\right),\quad i=1,\dots,3,
\end{equation}
where
\begin{equation}
\A=\left(\begin{array}{ccc}
\frac{a_{1} a_{2} a_{3}-a_{1} b_{2} c_{2}-a_{2} b_{2} c_{1}-a_{3} b_{1} c_{1}+b_{1} b_{2} c_{1}+b_{2} c_{1} c_{2}}{\left(a_{3}-b_{2}\right) \left(a_{2}-b_{1}\right)} & \frac{a_{2} b_{2}+a_{3} b_{1}-b_{1} b_{2}-b_{2} c_{2}}{\left(a_{3}-b_{2}\right) \left(a_{2}-b_{1}\right)} & \frac{b_{2}}{a_{3}-b_{2}} 
\\[2mm]
 \frac{2 a_{1} a_{2} a_{3}-a_{1} a_{2} b_{2}-a_{1} a_{3} b_{1}+a_{1} b_{1} b_{2}-a_{1} b_{2} c_{2}-a_{2} a_{3} c_{1}+b_{2} c_{1} c_{2}}{\left(a_{3}-b_{2}\right) \left(a_{2}-b_{1}\right)} & \frac{a_{2} a_{3}-b_{2} c_{2}}{\left(a_{3}-b_{2}\right) \left(a_{2}-b_{1}\right)} & \frac{b_{2}}{a_{3}-b_{2}} 
\\[2mm]
 \frac{2 a_{1} a_{2} a_{3}-a_{1} a_{2} b_{2}-a_{1} a_{3} c_{2}-a_{2} a_{3} c_{1}-a_{3} b_{1} c_{1}+a_{3} c_{1} c_{2}+b_{1} b_{2} c_{1}}{\left(a_{3}-b_{2}\right) \left(a_{2}-b_{1}\right)} & \frac{2 a_{2} a_{3}-a_{2} b_{2}-a_{3} c_{2}}{\left(a_{3}-b_{2}\right) \left(a_{2}-b_{1}\right)} & \frac{a_{3}}{a_{3}-b_{2}} 
\end{array}\right),
\end{equation}
\end{subequations}
which satisfies the condition $C_{1,2,3}=0$.
The LV system \eqref{yn3} admits the DPs
\[
 (a_{1}- c_{1})y_{1}+y_{2},\quad \left(a_{1}-c_{1}\right) \left(a_2-c_{2}\right)y_{1} + \left(a_{2}-c_{2}\right)y_{2} +\left(a_{2}-b_{1}\right) y_{3},
 \]
and corresponds to the other hypergraph of size 2 in Fig. \ref{nih3}. This hypergraph is therefore LV-equivalent to the tree on 3 vertices.
\item The two nonisomorphic hypergraphs of size $3$ are also LV-equivalent. Their LV-systems are subsystems of the $k=2$ systems, Eqs. \eqref{xn3} and \eqref{yn3}. In each case, the matrix $\A$ satisfies one extra condition and has 6 free parameters. We note that although the LV-system has three DPs, there are only two functionally independent integrals.
\item As we saw in Example \ref{h34}, the hypergraph with $k=4$ is not admissible.
\end{enumerate}
The number of hypergraphs of order 3 and size $k>0$, as well as how many of are admissible and nonequivalent is given in Table \ref{t0}. There are 4 nonequivalent 3-component LV-systems with additional DPs.

\begin{table}[H]
\begin{center}
\begin{tabular}{c|cccc|c}
$k$ & 1 & 2 & 3 & 4 & total \\
\hline
hypergraphs & 2 & 2 & 2 & 1 & 7 \\
admissible & 2 & 2 & 2 & 0 & 6\\
nonequivalent & 2 & 1 & 1 & 0 & 4\\
\end{tabular}
\vspace{3mm}

\parbox{11.5cm}{\caption{\label{t0} The number of order 3 hypergraphs, admissible hypergraphs, and equivalence classes of admissible hypergraphs, with $k>0$ hyperedges.}}
\end{center}
\end{table}

\subsection{4-component LV systems}
The number of hypergraphs of order 4 and size $k>0$, as well as how many of them are admissible and nonequivalent is given in Table \ref{t1}.

\begin{table}[H]
\begin{center}
\begin{tabular}{c|ccccccccccc|c}
$k$ & 1 & 2 & 3 & 4 & 5 & 6 & 7 & 8 & 9 & 10 & 11 & total \\
\hline
hypergraphs  & 3 & 7 & 16 & 28 & 35 & 35 & 28 & 16 & 7 & 3 & 1 & 179 \\
admissible  & 3 & 6 & 9 & 9 & 0 & 3 & 0 & 0 & 0 & 0 & 0 & 30\\
nonequivalent  & 3 & 3 & 3 & 1 & 0 & 1 & 0 & 0 & 0 & 0 & 0 & 11\\
\end{tabular}
\vspace{3mm}

\parbox{11.5cm}{\caption{\label{t1} The number of order 4 hypergraphs, admissible hypergraphs, and equivalence classes of admissible hypergraphs, with $k>0$ hyperedges.}}
\end{center}
\end{table}
\vspace{-12pt}
 
\begin{enumerate}[1.]
\item Nonisomorphic hypergraphs of size 1 are not LV-equivalent. For $n=4$, the size 1 hypergraphs are given in Fig. \ref{h4k1}.

\begin{figure}[H]
\centering
\begin{tikzpicture}[scale=1]
\tikzstyle{nod}= [circle, inner sep=0pt, fill=white, minimum size=2pt, draw]		
\node[nod] (a) at (0,0) {};
\node[nod] (b) at (1,0) {};
\node[nod] (c) at (0,1) {};
\node[nod] (d) at (1,1) {};
\draw[line width=1pt] (a) --  (c);
\def\r{3}	
\node[nod] (a) at (0+\r,0) {};
\node[nod] (b) at (1+\r,0) {};
\node[nod] (c) at (0+\r,1) {};
\node[nod] (d) at (1+\r,1) {};
\draw[line width=1pt] (0.1+\r,0.1) --  (.1+\r,.9) -- (.9+\r,.1) -- (.1+\r,.1);	
\def\r{6}
\node[nod] (a) at (0+\r,0) {};
\node[nod] (b) at (1+\r,0) {};
\node[nod] (c) at (0+\r,1) {};
\node[nod] (d) at (1+\r,1) {};
\draw[line width=1pt] (-0.1+\r,-0.1) --  (-.1+\r,1.1) -- (1.1+\r,1.1) -- (1.1+\r,-.1) -- (-0.1+\r,-0.1);	
\end{tikzpicture}
\caption{\label{h4k1} Nonequivalent admissible hypergraphs of order 4 and size 1.}
\end{figure}
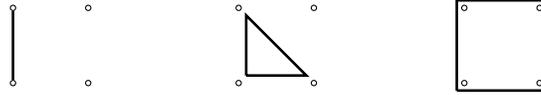

Their LV-matrices, with respectively 14,13 and 13 parameters, are given by
\[
\begin{pmatrix}
a_{1} & a_{2} & a_{5} & a_{6} 
\\
 a_{3} & a_{4} & a_{5} & a_{6} 
\\
 a_{7} & a_{8} & a_{9} & a_{10} 
\\
 a_{11} & a_{12} & a_{13} & a_{14} 
\end{pmatrix},\qquad
\begin{pmatrix}
a_{1} & a_{2} & a_{3} & a_{4} 
\\
 a_{5} & a_{6} & a_{7} & a_{4} 
\\
 a_{8} & a_{6}-\frac{\left(a_{2}-a_{6}\right) \left(a_{1}-a_{8}\right) \left(a_{7}-a_{9}\right)}{\left(a_{1}-a_{5}\right) \left(a_{3}-a_{9}\right)} & a_{9} & a_{4} 
\\
 a_{10} & a_{11} & a_{12} & a_{13} 
\end{pmatrix},
\]
and
\[
\begin{pmatrix}
a_{1} & a_{2} & a_{3} & a_{4} 
\\
 a_{5} & a_{6} & a_{9}+\frac{\left(a_{3}-a_{9}\right) \left(a_{1}-a_{5}\right) \left(a_{6}-a_{8}\right)}{\left(a_{2}-a_{6}\right) \left(a_{1}-a_{7}\right)} & a_{13}+\frac{\left(a_{4}-a_{13}\right) \left(a_{1}-a_{5}\right) \left(a_{6}-a_{12}\right)}{\left(a_{2}-a_{6}\right) \left(a_{1}-a_{11}\right)} 
\\
 a_{7} & a_{8} & a_{9} & a_{10} 
\\
 a_{11} & a_{12} & a_{9}-\frac{\left(a_{3}-a_{9}\right) \left(a_{1}-a_{11}\right) \left(a_{10}-a_{13}\right)}{\left(a_{1}-a_{7}\right) \left(a_{4}-a_{13}\right)} & a_{13} 
\end{pmatrix}.
\]
\item 
The size 2 LV-admissible hypergraphs are depicted in Fig. \ref{h4k2}. There are 3 equivalence classes which contain respectively 1,2 and 3 nonisomorphic hypergraphs.

\begin{figure}[H]
\centering
\begin{tikzpicture}[scale=1]
\tikzstyle{nod}= [circle, inner sep=0pt, fill=white, minimum size=2pt, draw]		
\def\r{6};
\node[nod] (a) at (0+\r,0) {};
\node[nod] (b) at (1+\r,0) {};
\node[nod] (c) at (0+\r,1) {};
\node[nod] (d) at (1+\r,1) {};
\draw[line width=1pt] (0.1+\r,0.1) --  (.1+\r,.9) -- (.9+\r,.1) -- (.1+\r,.1);
\draw[line width=1pt] (-0.1+\r,-0.1) --  (-.1+\r,1.1) -- (1.1+\r,1.1) -- (1.1+\r,-.1) -- (-0.1+\r,-0.1);
\def\r{3};
\node[nod] (a) at (0+\r,0) {};
\node[nod] (b) at (1+\r,0) {};
\node[nod] (c) at (0+\r,1) {};
\node[nod] (d) at (1+\r,1) {};
\draw[line width=1pt] (c) --  (d);
\draw[line width=1pt] (-0.1+\r,-0.1) --  (-.1+\r,1.1) -- (1.1+\r,1.1) -- (1.1+\r,-.1) -- (-0.1+\r,-0.1);
\def\r{0};
\node[nod] (a) at (0+\r,0) {};
\node[nod] (b) at (1+\r,0) {};
\node[nod] (c) at (0+\r,1) {};
\node[nod] (d) at (1+\r,1) {};
\draw[line width=1pt] (c) --  (d);
\draw[line width=1pt] (0.1+\r,0.1) --  (.1+\r,.9) -- (.9+\r,.1) -- (.1+\r,.1);
\def\r{3};
\def\u{2};
\node[nod] (a) at (0+\r,0+\u) {};
\node[nod] (b) at (1+\r,0+\u) {};
\node[nod] (c) at (0+\r,1+\u) {};
\node[nod] (d) at (1+\r,1+\u) {};
\draw[line width=1pt] (a) --  (b);
\draw[line width=1pt] (0.1+\r,0.1+\u) --  (.1+\r,.9+\u) -- (.9+\r,.1+\u) -- (.1+\r,.1+\u);
\def\r{0};
\def\u{2};
\node[nod] (a) at (0+\r,0+\u) {};
\node[nod] (b) at (1+\r,0+\u) {};
\node[nod] (c) at (0+\r,1+\u) {};
\node[nod] (d) at (1+\r,1+\u) {};
\draw[line width=1pt] (a) --  (b);
\draw[line width=1pt] (a) --  (c);
\def\r{0};
\def\u{4};
\node[nod] (a) at (0+\r,0+\u) {};
\node[nod] (b) at (1+\r,0+\u) {};
\node[nod] (c) at (0+\r,1+\u) {};
\node[nod] (d) at (1+\r,1+\u) {};
\draw[line width=1pt] (a) --  (c);
\draw[line width=1pt] (b) --  (d);
\end{tikzpicture}
\caption{\label{h4k2} Admissible hypergraphs of order 4 and size 2. Hypergraphs next to each other are LV-equivalent.}
\end{figure}
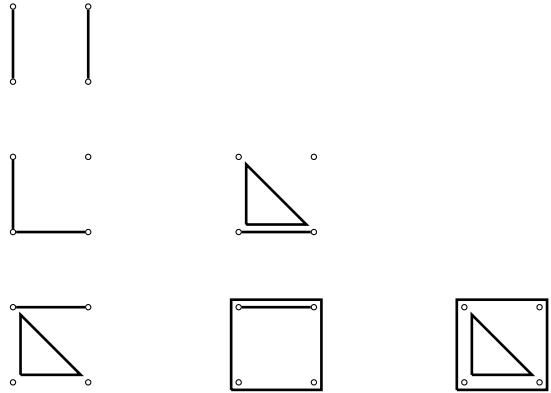

We take the hypergraphs whose edges are of lowest degree as representatives (these are the ones on the left in Fig. \ref{h4k2}). Their LV-matrices, with respectively 12, 12 and 11 parameters, are
\[
\begin{pmatrix}
a_{1} & a_{2} & a_{5} & a_{6} 
\\
 a_{3} & a_{4} & a_{5} & a_{6} 
\\
 a_{9} & a_{10} & a_{7} & a_{8} 
\\
 a_{9} & a_{10} & a_{11} & a_{12} 
\end{pmatrix},\qquad
\begin{pmatrix}
a_{1} & a_{6} & a_{4} & a_{8} 
\\
 a_{2} & a_{3} & a_{4} & a_{8} 
\\
 a_{5} & a_{6} & a_{7} & a_{8} 
\\
 a_{9} & a_{10} & a_{11} & a_{12} 
\end{pmatrix}
\]
and
\[
\begin{pmatrix}
a_{1} & a_{7} & a_{3} & a_{4} 
\\
 a_{5} & a_{6} & a_{3} & a_{4} 
\\
 a_{9} & a_{7} & a_{10} & a_{11}+\frac{\left(a_{4}-a_{11}\right) \left(a_{1}-a_{9}\right) \left(a_{10}-a_{8}\right)}{\left(a_{3}-a_{10}\right) \left(a_{1}-a_{2}\right)} 
\\
 a_{2} & a_{7} & a_{8} & a_{11} 
\end{pmatrix}.
\]
\item The size 3 LV-admissible hypergraphs are depicted in Fig. \ref{h4k3}. There are 3 equivalence classes which contain respectively 2,2 and 5 nonisomorphic hypergraphs.

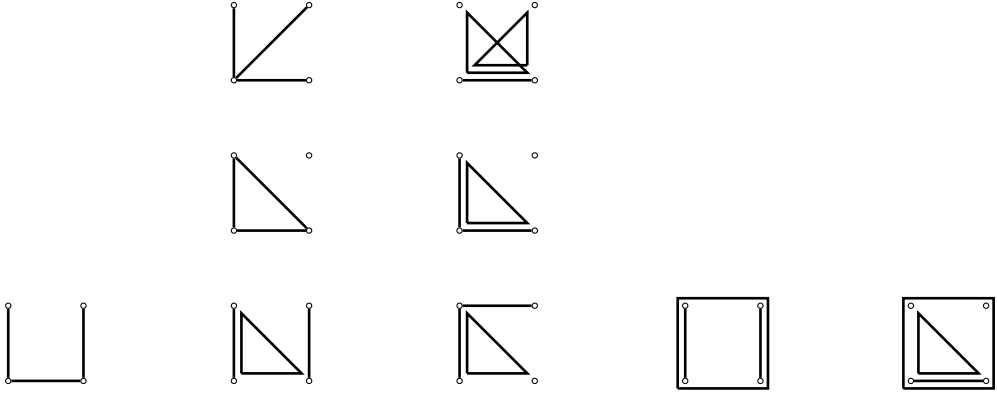
\begin{figure}[h]
\centering
\begin{tikzpicture}[scale=1]
\tikzstyle{nod}= [circle, inner sep=0pt, fill=white, minimum size=2pt, draw]		
\node[nod] (a) at (0,0) {};
\node[nod] (b) at (1,0) {};
\node[nod] (c) at (0,1) {};
\node[nod] (d) at (1,1) {};
\draw[line width=1pt] (c) --  (a) -- (b) -- (d);
\def\r{3};
\node[nod] (a) at (0+\r,0) {};
\node[nod] (b) at (1+\r,0) {};
\node[nod] (c) at (0+\r,1) {};
\node[nod] (d) at (1+\r,1) {};
\draw[line width=1pt] (c) --  (a);
\draw[line width=1pt] (b) --  (d);
\draw[line width=1pt] (0.1+\r,0.1) --  (.1+\r,.9) -- (.9+\r,.1) -- (.1+\r,.1);
\def\r{6};
\node[nod] (a) at (0+\r,0) {};
\node[nod] (b) at (1+\r,0) {};
\node[nod] (c) at (0+\r,1) {};
\node[nod] (d) at (1+\r,1) {};
\draw[line width=1pt] (c) --  (a);
\draw[line width=1pt] (c) --  (d);
\draw[line width=1pt] (0.1+\r,0.1) --  (.1+\r,.9) -- (.9+\r,.1) -- (.1+\r,.1);
\def\r{9};
\node[nod] (a) at (0+\r,0) {};
\node[nod] (b) at (1+\r,0) {};
\node[nod] (c) at (0+\r,1) {};
\node[nod] (d) at (1+\r,1) {};
\draw[line width=1pt] (c) --  (a);
\draw[line width=1pt] (b) --  (d);
\draw[line width=1pt] (-0.1+\r,-0.1) --  (-.1+\r,1.1) -- (1.1+\r,1.1) -- (1.1+\r,-.1) -- (-0.1+\r,-0.1);
\def\r{12};
\node[nod] (a) at (0+\r,0) {};
\node[nod] (b) at (1+\r,0) {};
\node[nod] (c) at (0+\r,1) {};
\node[nod] (d) at (1+\r,1) {};
\draw[line width=1pt] (b) --  (a);
\draw[line width=1pt] (-0.1+\r,-0.1) --  (-.1+\r,1.1) -- (1.1+\r,1.1) -- (1.1+\r,-.1) -- (-0.1+\r,-0.1);
\draw[line width=1pt] (0.1+\r,0.1) --  (.1+\r,.9) -- (.9+\r,.1) -- (.1+\r,.1);
\def\r{3};
\def\u{2};
\node[nod] (a) at (0+\r,0+\u) {};
\node[nod] (b) at (1+\r,0+\u) {};
\node[nod] (c) at (0+\r,1+\u) {};
\node[nod] (d) at (1+\r,1+\u) {};
\draw[line width=1pt] (a) --  (b) --  (c) -- (a);
\def\r{6};
\def\u{2};
\node[nod] (a) at (0+\r,0+\u) {};
\node[nod] (b) at (1+\r,0+\u) {};
\node[nod] (c) at (0+\r,1+\u) {};
\node[nod] (d) at (1+\r,1+\u) {};
\draw[line width=1pt] (c) -- (a) --  (b);
\draw[line width=1pt] (0.1+\r,0.1+\u) --  (.1+\r,.9+\u) -- (.9+\r,.1+\u) -- (.1+\r,.1+\u);
\def\r{3};
\def\u{4};
\node[nod] (a) at (0+\r,0+\u) {};
\node[nod] (b) at (1+\r,0+\u) {};
\node[nod] (c) at (0+\r,1+\u) {};
\node[nod] (d) at (1+\r,1+\u) {};
\draw[line width=1pt] (c) -- (a) --  (b);
\draw[line width=1pt] (a) --  (d);
\def\r{6};
\def\u{4};
\node[nod] (a) at (0+\r,0+\u) {};
\node[nod] (b) at (1+\r,0+\u) {};
\node[nod] (c) at (0+\r,1+\u) {};
\node[nod] (d) at (1+\r,1+\u) {};
\draw[line width=1pt] (a) --  (b);
\draw[line width=1pt] (0.1+\r,0.1+\u) --  (.1+\r,.9+\u) -- (.9+\r,.1+\u) -- (.1+\r,.1+\u);
\draw[line width=1pt] (.9+\r,.2+\u) -- (.2+\r,.2+\u) --  (.9+\r,.9+\u) -- (.9+\r,.2+\u);
\end{tikzpicture}
\caption{\label{h4k3} Admissible hypergraphs on 4 vertices with 3 hyperedges. Hypergraphs next to each other are LV-equivalent.}
\end{figure}

Representative LV-matrices, with respectively 10, 11 and 10 parameters, are
\begin{equation} \label{lvm43}
\begin{pmatrix}
a_{1} & a_{8} & a_{9} & a_{6} 
\\
 a_{2} & a_{3} & a_{9} & a_{6} 
\\
 a_{4} & a_{8} & a_{5} & a_{6} 
\\
 a_{7} & a_{8} & a_{9} & a_{10} 
\end{pmatrix},\qquad
\begin{pmatrix}
a_{1} & a_{5} & a_{3} & a_{7} 
\\
 a_{4} & a_{2} & a_{3} & a_{7} 
\\
 a_{4} & a_{5} & a_{6} & a_{7} 
\\
 a_{8} & a_{9} & a_{10} & a_{11} 
\end{pmatrix},\qquad
\begin{pmatrix}
a_{1} & a_{4} & a_{9} & a_{6} 
\\
 a_{7} & a_{2} & a_{9} & a_{6} 
\\
 a_{3} & a_{4} & a_{5} & a_{6} 
\\
 a_{7} & a_{8} & a_{9} & a_{10} 
\end{pmatrix}.
\end{equation}
For the second LV-system, there are only 2 functionally independent integrals. The other LV-systems are tree-systems. 
\item The size 4 LV-admissible hypergraphs are depicted in Fig. \ref{h4k4}. There is 1 equivalence class which contains 9 nonisomorphic hypergraphs.

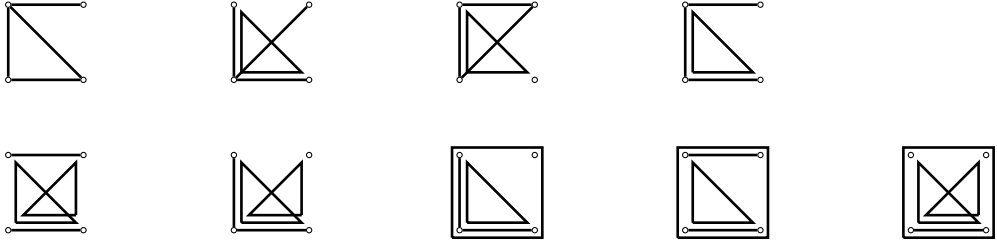
\begin{figure}[h]
\centering
\begin{tikzpicture}[scale=1]
\tikzstyle{nod}= [circle, inner sep=0pt, fill=white, minimum size=2pt, draw]		
\node[nod] (a) at (0,0) {};
\node[nod] (b) at (1,0) {};
\node[nod] (c) at (0,1) {};
\node[nod] (d) at (1,1) {};
\draw[line width=1pt] (c) -- (a) --  (b) -- (c) -- (d);
\def\r{3};
\node[nod] (a) at (0+\r,0) {};
\node[nod] (b) at (1+\r,0) {};
\node[nod] (c) at (0+\r,1) {};
\node[nod] (d) at (1+\r,1) {};
\draw[line width=1pt] (c) --  (a) -- (b);
\draw[line width=1pt] (a) --  (d);
\draw[line width=1pt] (0.1+\r,0.1) --  (.1+\r,.9) -- (.9+\r,.1) -- (.1+\r,.1);
\def\r{6};
\node[nod] (a) at (0+\r,0) {};
\node[nod] (b) at (1+\r,0) {};
\node[nod] (c) at (0+\r,1) {};
\node[nod] (d) at (1+\r,1) {};
\draw[line width=1pt] (a) -- (c) -- (d) -- (a);
\draw[line width=1pt] (0.1+\r,0.1) --  (.1+\r,.9) -- (.9+\r,.1) -- (.1+\r,.1);
\def\r{9};
\node[nod] (a) at (0+\r,0) {};
\node[nod] (b) at (1+\r,0) {};
\node[nod] (c) at (0+\r,1) {};
\node[nod] (d) at (1+\r,1) {};
\draw[line width=1pt] (b) -- (a) -- (c) -- (d);
\draw[line width=1pt] (0.1+\r,0.1) --  (.1+\r,.9) -- (.9+\r,.1) -- (.1+\r,.1);
\def\r{0};
\def\u{-2};
\node[nod] (a) at (0+\r,0+\u) {};
\node[nod] (b) at (1+\r,0+\u) {};
\node[nod] (c) at (0+\r,1+\u) {};
\node[nod] (d) at (1+\r,1+\u) {};
\draw[line width=1pt] (a) --  (b);
\draw[line width=1pt] (c) --  (d);
\draw[line width=1pt] (0.1+\r,0.1+\u) --  (.1+\r,.9+\u) -- (.9+\r,.1+\u) -- (.1+\r,.1+\u);
\draw[line width=1pt] (.9+\r,.2+\u) -- (.2+\r,.2+\u) --  (.9+\r,.9+\u) -- (.9+\r,.2+\u);
\def\r{3};
\def\u{-2};
\node[nod] (a) at (0+\r,0+\u) {};
\node[nod] (b) at (1+\r,0+\u) {};
\node[nod] (c) at (0+\r,1+\u) {};
\node[nod] (d) at (1+\r,1+\u) {};
\draw[line width=1pt] (c) --  (a) --  (b);
\draw[line width=1pt] (0.1+\r,0.1+\u) --  (.1+\r,.9+\u) -- (.9+\r,.1+\u) -- (.1+\r,.1+\u);
\draw[line width=1pt] (.9+\r,.2+\u) -- (.2+\r,.2+\u) --  (.9+\r,.9+\u) -- (.9+\r,.2+\u);
\def\r{6};
\def\u{-2};
\node[nod] (a) at (0+\r,0+\u) {};
\node[nod] (b) at (1+\r,0+\u) {};
\node[nod] (c) at (0+\r,1+\u) {};
\node[nod] (d) at (1+\r,1+\u) {};
\draw[line width=1pt] (b) --  (a) --  (c);
\draw[line width=1pt] (0.1+\r,0.1+\u) --  (.1+\r,.9+\u) -- (.9+\r,.1+\u) -- (.1+\r,.1+\u);
\draw[line width=1pt] (-0.1+\r,-0.1+\u) --  (-.1+\r,1.1+\u) -- (1.1+\r,1.1+\u) -- (1.1+\r,-.1+\u) -- (-0.1+\r,-0.1+\u);
\def\r{9};
\def\u{-2};
\node[nod] (a) at (0+\r,0+\u) {};
\node[nod] (b) at (1+\r,0+\u) {};
\node[nod] (c) at (0+\r,1+\u) {};
\node[nod] (d) at (1+\r,1+\u) {};
\draw[line width=1pt] (b) --  (a);
\draw[line width=1pt] (d) --  (c);
\draw[line width=1pt] (0.1+\r,0.1+\u) --  (.1+\r,.9+\u) -- (.9+\r,.1+\u) -- (.1+\r,.1+\u);
\draw[line width=1pt] (-0.1+\r,-0.1+\u) --  (-.1+\r,1.1+\u) -- (1.1+\r,1.1+\u) -- (1.1+\r,-.1+\u) -- (-0.1+\r,-0.1+\u);
\def\r{12};
\def\u{-2};
\node[nod] (a) at (0+\r,0+\u) {};
\node[nod] (b) at (1+\r,0+\u) {};
\node[nod] (c) at (0+\r,1+\u) {};
\node[nod] (d) at (1+\r,1+\u) {};
\draw[line width=1pt] (a) --  (b);
\draw[line width=1pt] (0.1+\r,0.1+\u) --  (.1+\r,.9+\u) -- (.9+\r,.1+\u) -- (.1+\r,.1+\u);
\draw[line width=1pt] (.9+\r,.2+\u) -- (.2+\r,.2+\u) --  (.9+\r,.9+\u) -- (.9+\r,.2+\u);
\draw[line width=1pt] (-0.1+\r,-0.1+\u) --  (-.1+\r,1.1+\u) -- (1.1+\r,1.1+\u) -- (1.1+\r,-.1+\u) -- (-0.1+\r,-0.1+\u);
\end{tikzpicture}
\caption{\label{h4k4} LV-equivalent admissible hypergraphs on 4 vertices with 4 hyperedges.}
\end{figure}

The simplest representative LV-matrix,
\[
\begin{pmatrix}
a_{1} & a_{7} & a_{8} & a_{5} 
\\
 a_{3} & a_{2} & a_{8} & a_{5} 
\\
 a_{3} & a_{7} & a_{4} & a_{5} 
\\
 a_{6} & a_{7} & a_{8} & a_{9} 

\end{pmatrix}
\]
has 9 parameters. It is (equivalent to) a subcase of each of the LV-systems with 3 DPs, cf. Eq. \eqref{lvm43}.
\item There are no size 5 LV-admissible hypergraphs of order 4.
\item There are 3 nonisomorphic LV-equivalent admissible hypergraphs of order 4 and size 6, cf. Fig. \ref{h4k6}.
     
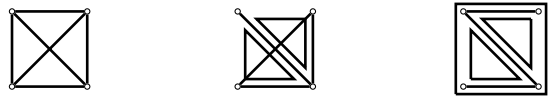
\begin{figure}[H]
\centering
\begin{tikzpicture}[scale=1]
\tikzstyle{nod}= [circle, inner sep=0pt, fill=white, minimum size=2pt, draw]		
\node[nod] (a) at (0,0) {};
\node[nod] (b) at (1,0) {};
\node[nod] (c) at (0,1) {};
\node[nod] (d) at (1,1) {};
\draw[line width=1pt] (a) --  (b) -- (c) -- (d) -- (a) -- (c);
\draw[line width=1pt] (b) -- (d);
\def\r{3};
\node[nod] (a) at (0+\r,0) {};
\node[nod] (b) at (1+\r,0) {};
\node[nod] (c) at (0+\r,1) {};
\node[nod] (d) at (1+\r,1) {};
\draw[line width=1pt] (a) -- (b) -- (d) -- (a);
\draw[line width=1pt] (b) --  (c);
\draw[line width=1pt] (0.1+\r,0.1) --  (.1+\r,.75) -- (.75+\r,.1) -- (.1+\r,.1);
\draw[line width=1pt] (0.9+\r,0.9) --  (.9+\r,.25) -- (.25+\r,.9) -- (.9+\r,.9);
\def\r{6};
\node[nod] (a) at (0+\r,0) {};
\node[nod] (b) at (1+\r,0) {};
\node[nod] (c) at (0+\r,1) {};
\node[nod] (d) at (1+\r,1) {};
\draw[line width=1pt] (a) -- (b) -- (c) -- (d);
\draw[line width=1pt] (0.1+\r,0.1) --  (.1+\r,.75) -- (.75+\r,.1) -- (.1+\r,.1);
\draw[line width=1pt] (0.9+\r,0.9) --  (.9+\r,.25) -- (.25+\r,.9) -- (.9+\r,.9);
\draw[line width=1pt] (-0.1+\r,-0.1) --  (-.1+\r,1.1) -- (1.1+\r,1.1) -- (1.1+\r,-.1) -- (-0.1+\r,-0.1);
\end{tikzpicture}
\caption{\label{h4k6} LV-equivalent admissible hypergraphs on 4 vertices with 6 hyperedges.}
\end{figure}
A representative LV-matrix, associated to the complete graph on 4 vertices, is
\[
\begin{pmatrix}
  a_{1} & a_{6} & a_{7} & a_{4} 
\\
 a_{5} & a_{2} & a_{7} & a_{4} 
\\
 a_{5} & a_{6} & a_{3} & a_{4} 
\\
 a_{5} & a_{6} & a_{7} & a_{8}  
\end{pmatrix},
\]
which has 8 parameters.
\end{enumerate}
\subsection{5-component LV systems}
The number of hypergraphs of order 5 and size $1\leq k\leq7$, as well as how many of them are admissible/nonequivalent is given in Table \ref{t2}.
\begin{table}[h]
\begin{center}
\begin{tabular}{c|cccccccccccc|c}
$k$ & 1 & 2 & 3 & 4 & 5 & 6 & 7 \\
\hline
hypergraphs & 4& 15& 62& 243& 841& 2544& 6672 \\
admissible & 4& 12& 27& 45& 60& 18& 21 \\
nonequivalent & 4& 6& 7& 7& 4& 2& 1 \\
\end{tabular}
\vspace{3mm}

\parbox{11.5cm}{\caption{\label{t2} The number of order 5 hypergraphs, admissible hypergraphs, and equivalence classes of admissible hypergraphs, with $1\leq k\leq 7$ hyperedges.}}
\end{center}
\end{table}
\vspace{-8pt} 

\begin{enumerate}[1.]
\item Nonisomorphic hypergraphs of size 1 are not LV-equivalent. For $n=5$, the size 1 hypergraphs are given in Fig. \ref{h5k1}.
\begin{figure}[H]
\centering
\begin{tikzpicture}[scale=.8]
\tikzstyle{nod}= [circle, inner sep=0pt, fill=white, minimum size=2pt, draw]		
\def\x{0}
\def\y{0}
\def\s{1}
\node[nod] (a) at (\s*1+\x,0+\y) {};
\node[nod] (b) at (\s*.31+\x,\s*.95+\y) {};
\node[nod] (c) at (-\s*.81+\x,\s*.59+\y) {};
\node[nod] (d) at (-\s*.81+\x,-\s*.59+\y) {};
\node[nod] (e) at (\s*.31+\x,-\s*.95+\y) {};
\draw[line width=1pt] (c) -- (d);
\def\x{3}
\def\y{0}
\def\s{1}
\node[nod] (a) at (\s*1+\x,0+\y) {};
\node[nod] (b) at (\s*.31+\x,\s*.95+\y) {};
\node[nod] (c) at (-\s*.81+\x,\s*.59+\y) {};
\node[nod] (d) at (-\s*.81+\x,-\s*.59+\y) {};
\node[nod] (e) at (\s*.31+\x,-\s*.95+\y) {};
\def\s{.9}
\coordinate (a) at (\s*1+\x,0+\y);
\coordinate (b) at (\s*.31+\x,\s*.95+\y);
\coordinate (c) at (-\s*.81+\x,\s*.59+\y);
\coordinate (d) at (-\s*.81+\x,-\s*.59+\y);
\coordinate (e) at (\s*.31+\x,-\s*.95+\y);
\draw[line width=.8pt] (a) -- (c) -- (d) -- (a);
\def\x{6}
\def\y{0}
\def\s{1}
\node[nod] (a) at (\s*1+\x,0+\y) {};
\node[nod] (b) at (\s*.31+\x,\s*.95+\y) {};
\node[nod] (c) at (-\s*.81+\x,\s*.59+\y) {};
\node[nod] (d) at (-\s*.81+\x,-\s*.59+\y) {};
\node[nod] (e) at (\s*.31+\x,-\s*.95+\y) {};
\def\s{.9}
\coordinate (a) at (\s*1+\x,0+\y);
\coordinate (b) at (\s*.31+\x,\s*.95+\y);
\coordinate (c) at (-\s*.81+\x,\s*.59+\y);
\coordinate (d) at (-\s*.81+\x,-\s*.59+\y);
\coordinate (e) at (\s*.31+\x,-\s*.95+\y);
\draw[line width=.8pt] (b) -- (c) -- (d) -- (e) -- (b);
\def\x{9}
\def\y{0}
\def\s{1}
\node[nod] (a) at (\s*1+\x,0+\y) {};
\node[nod] (b) at (\s*.31+\x,\s*.95+\y) {};
\node[nod] (c) at (-\s*.81+\x,\s*.59+\y) {};
\node[nod] (d) at (-\s*.81+\x,-\s*.59+\y) {};
\node[nod] (e) at (\s*.31+\x,-\s*.95+\y) {};
\def\s{.9}
\coordinate (a) at (\s*1+\x,0+\y);
\coordinate (b) at (\s*.31+\x,\s*.95+\y);
\coordinate (c) at (-\s*.81+\x,\s*.59+\y);
\coordinate (d) at (-\s*.81+\x,-\s*.59+\y);
\coordinate (e) at (\s*.31+\x,-\s*.95+\y);
\draw[line width=.8pt] (a) -- (b) -- (c) -- (d) -- (e) -- (a);
\end{tikzpicture}
\caption{\label{h5k1} Nonequivalent admissible hypergraphs on 5 vertices with 1 hyperedge.}
\end{figure}
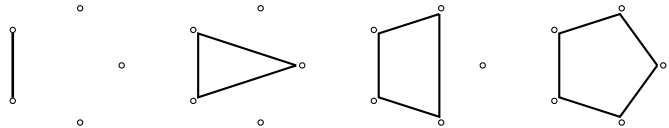

with respectively 22,20,19 and 19 parameters, are given by
\[
\begin{pmatrix}
a_{1} & a_{2} & a_{5} & a_{6} & a_{7} 
\\
 a_{3} & a_{4} & a_{5} & a_{6} & a_{7} 
\\
 a_{8} & a_{9} & a_{10} & a_{11} & a_{12} 
\\
 a_{13} & a_{14} & a_{15} & a_{16} & a_{17} 
\\
 a_{18} & a_{19} & a_{20} & a_{21} & a_{22}
\end{pmatrix},\qquad
\begin{pmatrix}
a_{1} & a_{2} & a_{3} & a_{9} & a_{10} 
\\
 a_{4} & a_{5} & a_{6} & a_{9} & a_{10} 
\\
 a_{1}-\frac{\left(a_{5}-a_{7}\right) \left(a_{3}-a_{8}\right) \left(a_{1}-a_{4}\right)}{\left(a_{2}-a_{5}\right) \left(a_{6}-a_{8}\right)} & a_{7} & a_{8} & a_{9} & a_{10} 
\\
 a_{11} & a_{12} & a_{13} & a_{14} & a_{15} 
\\
 a_{16} & a_{17} & a_{18} & a_{19} & a_{20} 
\end{pmatrix},
\]
\[
\begin{pmatrix}
a_{1} & a_{2} & a_{3} & a_{4} & a_{14} 
\\
 a_{5} & a_{6} & a_{9}+\frac{\left(a_{6}-a_{8}\right) \left(a_{3}-a_{9}\right) \left(a_{1}-a_{5}\right)}{\left(a_{2}-a_{6}\right) \left(a_{1}-a_{7}\right)} & a_{13}+\frac{\left(a_{1}-a_{5}\right) \left(a_{4}-a_{13}\right) \left(a_{6}-a_{11}\right)}{\left(a_{2}-a_{6}\right) \left(a_{1}-a_{10}\right)} & a_{14} 
\\
 a_{7} & a_{8} & a_{9} & a_{13}+\frac{\left(a_{1}-a_{7}\right) \left(a_{4}-a_{13}\right) \left(a_{9}-a_{12}\right)}{\left(a_{1}-a_{10}\right) \left(a_{3}-a_{9}\right)} & a_{14} 
\\
 a_{10} & a_{11} & a_{12} & a_{13} & a_{14} 
\\
 a_{15} & a_{16} & a_{17} & a_{18} & a_{19} 
\end{pmatrix},
\]
\[
\begin{pmatrix}
a_{1} & a_{2} & a_{3} & a_{4} & a_{5} 
\\
 a_{6} & a_{7} & a_{10}+\frac{\left(a_{7}-a_{9}\right) \left(a_{1}-a_{6}\right) \left(a_{3}-a_{10}\right)}{\left(a_{2}-a_{7}\right) \left(a_{1}-a_{8}\right)} & a_{14}+\frac{\left(a_{1}-a_{6}\right) \left(a_{4}-a_{14}\right) \left(a_{7}-a_{12}\right)}{\left(a_{2}-a_{7}\right) \left(-a_{11}+a_{1}\right)} & a_{19}+\frac{\left(a_{1}-a_{6}\right) \left(a_{5}-a_{19}\right) \left(a_{7}-a_{16}\right)}{\left(a_{2}-a_{7}\right) \left(a_{1}-a_{15}\right)} 
\\
 a_{8} & a_{9} & a_{10} & a_{14}+\frac{\left(a_{1}-a_{8}\right) \left(a_{4}-a_{14}\right) \left(a_{10}-a_{13}\right)}{\left(a_{1}-a_{11}\right) \left(a_{3}-a_{10}\right)} & a_{19}+\frac{\left(a_{1}-a_{8}\right) \left(a_{5}-a_{19}\right) \left(a_{10}-a_{17}\right)}{\left(a_{1}-a_{15}\right) \left(a_{3}-a_{10}\right)} 
\\
 a_{11} & a_{12} & a_{13} & a_{14} & a_{19}+\frac{\left(a_{5}-a_{19}\right) \left(a_{14}-a_{18}\right) \left(a_{1}-a_{11}\right)}{\left(a_{1}-a_{15}\right) \left(a_{4}-a_{14}\right)} 
\\
 a_{15} & a_{16} & a_{17} & a_{18} & a_{19} 
\end{pmatrix}.
\]

\item Representatives of the size 2 admissible hypergraph LV-equivalence classes are depicted in Fig. \ref{h5k2}.
\begin{figure}[H]
\centering
\begin{tikzpicture}[scale=.8]
\tikzstyle{nod}= [circle, inner sep=0pt, fill=white, minimum size=2pt, draw]		
\def\x{0}
\def\y{0}
\def\s{1}
\node[nod] (a) at (\s*1+\x,0+\y) {};
\node[nod] (b) at (\s*.31+\x,\s*.95+\y) {};
\node[nod] (c) at (-\s*.81+\x,\s*.59+\y) {};
\node[nod] (d) at (-\s*.81+\x,-\s*.59+\y) {};
\node[nod] (e) at (\s*.31+\x,-\s*.95+\y) {};
\draw[line width=1pt] (b) -- (c) -- (d);
\node (o) at (\x,0) {2};
\def\x{3}
\def\y{0}
\def\s{1}
\node[nod] (a) at (\s*1+\x,0+\y) {};
\node[nod] (b) at (\s*.31+\x,\s*.95+\y) {};
\node[nod] (c) at (-\s*.81+\x,\s*.59+\y) {};
\node[nod] (d) at (-\s*.81+\x,-\s*.59+\y) {};
\node[nod] (e) at (\s*.31+\x,-\s*.95+\y) {};
\draw[line width=1pt] (b) -- (a);
\draw[line width=1pt] (c) -- (d);
\node (o) at (\x,0) {1};
\def\x{6}
\def\y{0}
\def\s{1}
\node[nod] (a) at (\s*1+\x,0+\y) {};
\node[nod] (b) at (\s*.31+\x,\s*.95+\y) {};
\node[nod] (c) at (-\s*.81+\x,\s*.59+\y) {};
\node[nod] (d) at (-\s*.81+\x,-\s*.59+\y) {};
\node[nod] (e) at (\s*.31+\x,-\s*.95+\y) {};
\draw[line width=1pt] (c) -- (d);
\def\s{.9}
\coordinate (a) at (\s*1+\x,0+\y);
\coordinate (b) at (\s*.31+\x,\s*.95+\y);
\coordinate (c) at (-\s*.81+\x,\s*.59+\y);
\coordinate (d) at (-\s*.81+\x,-\s*.59+\y);
\coordinate (e) at (\s*.31+\x,-\s*.95+\y);
\draw[line width=.8pt] (a) -- (d) -- (e) -- (a);
\node (o) at (\x,.2) {3};
\def\x{9}
\def\y{0}
\def\s{1}
\node[nod] (a) at (\s*1+\x,0+\y) {};
\node[nod] (b) at (\s*.31+\x,\s*.95+\y) {};
\node[nod] (c) at (-\s*.81+\x,\s*.59+\y) {};
\node[nod] (d) at (-\s*.81+\x,-\s*.59+\y) {};
\node[nod] (e) at (\s*.31+\x,-\s*.95+\y) {};
\draw[line width=1pt] (c) -- (d);
\def\s{.9}
\coordinate (a) at (\s*1+\x,0+\y);
\coordinate (b) at (\s*.31+\x,\s*.95+\y);
\coordinate (c) at (-\s*.81+\x,\s*.59+\y);
\coordinate (d) at (-\s*.81+\x,-\s*.59+\y);
\coordinate (e) at (\s*.31+\x,-\s*.95+\y);
\draw[line width=.8pt] (a) -- (b) -- (e) -- (a);
\node (o) at (\x-.2,0) {1};
\def\x{12}
\def\y{0}
\def\s{1}
\node[nod] (a) at (\s*1+\x,0+\y) {};
\node[nod] (b) at (\s*.31+\x,\s*.95+\y) {};
\node[nod] (c) at (-\s*.81+\x,\s*.59+\y) {};
\node[nod] (d) at (-\s*.81+\x,-\s*.59+\y) {};
\node[nod] (e) at (\s*.31+\x,-\s*.95+\y) {};
\draw[line width=1pt] (c) -- (d);
\def\s{.9}
\coordinate (a) at (\s*1+\x,0+\y);
\coordinate (b) at (\s*.31+\x,\s*.95+\y);
\coordinate (c) at (-\s*.81+\x,\s*.59+\y);
\coordinate (d) at (-\s*.81+\x,-\s*.59+\y);
\coordinate (e) at (\s*.31+\x,-\s*.95+\y);
\draw[line width=.8pt] (a) -- (b) -- (d) -- (e) -- (a);
\node (o) at (\x+.1,-.1) {3};
\def\x{15}
\def\y{0}
\def\s{1}
\node[nod] (a) at (\s*1+\x,0+\y) {};
\node[nod] (b) at (\s*.31+\x,\s*.95+\y) {};
\node[nod] (c) at (-\s*.81+\x,\s*.59+\y) {};
\node[nod] (d) at (-\s*.81+\x,-\s*.59+\y) {};
\node[nod] (e) at (\s*.31+\x,-\s*.95+\y) {};
\def\s{.9}
\coordinate (a) at (\s*1+\x,0+\y);
\coordinate (b) at (\s*.31+\x,\s*.95+\y);
\coordinate (c) at (-\s*.81+\x,\s*.59+\y);
\coordinate (d) at (-\s*.81+\x,-\s*.59+\y);
\coordinate (e) at (\s*.31+\x,-\s*.95+\y);
\draw[line width=.8pt] (a) -- (d) -- (e) -- (a);
\draw[line width=.8pt] (b) -- (d) -- (c) -- (b);
\node (o) at (\x+.1,+.1) {2};
\end{tikzpicture}
\caption{\label{h5k2} Nonequivalent admissible hypergraphs on 5 vertices with 2 hyperedges. We have added the number of nonisomorphic hypergraphs that are LV-equivalent.}
\end{figure}
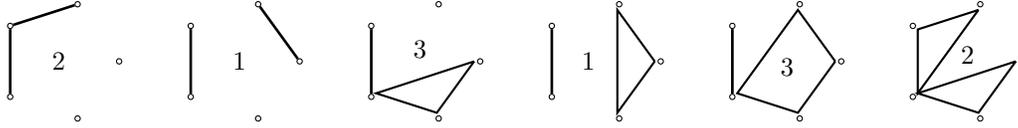
Corresponding LV-matrices, where the number of parameters is the index of the parameter in the bottom right corner, are
\[
\begin{pmatrix}
  a_{1} & a_{6} & a_{4} & a_{8} & a_{9} 
\\
 a_{2} & a_{3} & a_{4} & a_{8} & a_{9} 
\\
 a_{5} & a_{6} & a_{7} & a_{8} & a_{9} 
\\
 a_{10} & a_{11} & a_{12} & a_{13} & a_{14} 
\\
 a_{15} & a_{16} & a_{17} & a_{18} & a_{19}  
\end{pmatrix},\qquad 
\begin{pmatrix}
  a_{1} & a_{2} & a_{5} & a_{6} & a_{7} 
\\
 a_{3} & a_{4} & a_{5} & a_{6} & a_{7} 
\\
 a_{10} & a_{11} & a_{8} & a_{9} & a_{14} 
\\
 a_{10} & a_{11} & a_{12} & a_{13} & a_{14} 
\\
 a_{15} & a_{16} & a_{17} & a_{18} & a_{19} 
\end{pmatrix},
\]
\[
\begin{pmatrix}
  a_{1} & a_{9} & a_{2} & a_{3} & a_{12} 
\\
 a_{4} & a_{5} & a_{2} & a_{3} & a_{12} 
\\
 a_{1}-\frac{\left(a_{2}-a_{6}\right) \left(a_{1}-a_{8}\right) \left(a_{7}-a_{11}\right)}{\left(a_{3}-a_{11}\right) \left(a_{6}-a_{10}\right)} & a_{9} & a_{6} & a_{7} & a_{12} 
\\
 a_{8} & a_{9} & a_{10} & a_{11} & a_{12} 
\\
 a_{13} & a_{14} & a_{15} & a_{16} & a_{17} 
\end{pmatrix},
\]
\[
\begin{pmatrix}
a_{1} & a_{2} & a_{5} & a_{6} & a_{7} 
\\
 a_{3} & a_{4} & a_{5} & a_{6} & a_{7} 
\\
 a_{13} & a_{14} & a_{8} & a_{11}+\frac{\left(a_{9}-a_{17}\right) \left(a_{11}-a_{16}\right) \left(a_{8}-a_{10}\right)}{\left(a_{12}-a_{17}\right) \left(a_{8}-a_{15}\right)} & a_{9} 
\\
 a_{13} & a_{14} & a_{10} & a_{11} & a_{12} 
\\
 a_{13} & a_{14} & a_{15} & a_{16} & a_{17}    
\end{pmatrix},
\]
\[
\begin{pmatrix}
a_{1} & a_{13} & a_{2} & a_{3} & a_{4} 
\\
 a_{5} & a_{6} & a_{2} & a_{3} & a_{4} 
\\
 a_{7} & a_{13} & a_{8} & a_{11}+\frac{\left(a_{1}-a_{7}\right) \left(a_{3}-a_{11}\right) \left(a_{8}-a_{10}\right)}{\left(a_{1}-a_{9}\right) \left(a_{2}-a_{8}\right)} & a_{16}+\frac{\left(a_{1}-a_{7}\right) \left(a_{4}-a_{16}\right) \left(a_{8}-a_{14}\right)}{\left(a_{1}-a_{12}\right) \left(a_{2}-a_{8}\right)} 
\\
 a_{9} & a_{13} & a_{10} & a_{11} & a_{16}+\frac{\left(a_{1}-a_{9}\right) \left(a_{4}-a_{16}\right) \left(a_{11}-a_{15}\right)}{\left(a_{1}-a_{12}\right) \left(a_{3}-a_{11}\right)} 
\\
 a_{12} & a_{13} & a_{14} & a_{15} & a_{16}    
\end{pmatrix},
\]
\[
\begin{pmatrix}
 a_{1} & a_{2} & a_{10}+\frac{\left(a_{2}-a_{6}\right) \left(a_{1}-a_{8}\right) \left(a_{7}-a_{10}\right)}{\left(a_{6}-a_{9}\right) \left(a_{1}-a_{5}\right)} & a_{3} & a_{4} 
\\
 a_{5} & a_{6} & a_{7} & a_{3} & a_{4} 
\\
 a_{8} & a_{9} & a_{10} & a_{3} & a_{4} 
\\
 a_{11} & a_{2} & a_{10}+\frac{\left(a_{2}-a_{6}\right) \left(a_{1}-a_{8}\right) \left(a_{7}-a_{10}\right)}{\left(a_{6}-a_{9}\right) \left(a_{1}-a_{5}\right)} & a_{12} & a_{15}+\frac{\left(a_{4}-a_{15}\right) \left(a_{1}-a_{11}\right) \left(a_{12}-a_{14}\right)}{\left(a_{1}-a_{13}\right) \left(a_{3}-a_{12}\right)} 
\\
 a_{13} & a_{2} & a_{10}+\frac{\left(a_{2}-a_{6}\right) \left(a_{1}-a_{8}\right) \left(a_{7}-a_{10}\right)}{\left(a_{6}-a_{9}\right) \left(a_{1}-a_{5}\right)} & a_{14} & a_{15}   
\end{pmatrix}.
\]
\item Representatives of the size 3 admissible hypergraph LV-equivalence classes are depicted in Fig. \ref{h5k3}.
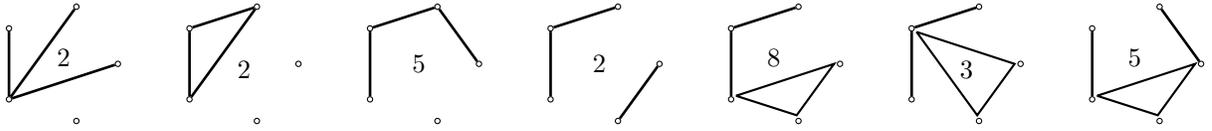
\begin{figure}[H]
\centering
\begin{tikzpicture}[scale=.8]
\tikzstyle{nod}= [circle, inner sep=0pt, fill=white, minimum size=2pt, draw]		
\def\x{0}
\def\y{3}
\def\s{1}
\node[nod] (a) at (\s*1+\x,0+\y) {};
\node[nod] (b) at (\s*.31+\x,\s*.95+\y) {};
\node[nod] (c) at (-\s*.81+\x,\s*.59+\y) {};
\node[nod] (d) at (-\s*.81+\x,-\s*.59+\y) {};
\node[nod] (e) at (\s*.31+\x,-\s*.95+\y) {};
\draw[line width=1pt] (c) -- (d) -- (b);
\draw[line width=1pt] (d) -- (a);
\node (o) at (\x+.1,\y+.1) {2};
\def\x{3}
\def\y{3}
\def\s{1}
\node[nod] (a) at (\s*1+\x,0+\y) {};
\node[nod] (b) at (\s*.31+\x,\s*.95+\y) {};
\node[nod] (c) at (-\s*.81+\x,\s*.59+\y) {};
\node[nod] (d) at (-\s*.81+\x,-\s*.59+\y) {};
\node[nod] (e) at (\s*.31+\x,-\s*.95+\y) {};
\draw[line width=1pt] (b) -- (c) -- (d) -- (b);
\node (o) at (\x+.1,\y-.1) {2};
\def\x{6}
\def\y{3}
\def\s{1}
\node[nod] (a) at (\s*1+\x,0+\y) {};
\node[nod] (b) at (\s*.31+\x,\s*.95+\y) {};
\node[nod] (c) at (-\s*.81+\x,\s*.59+\y) {};
\node[nod] (d) at (-\s*.81+\x,-\s*.59+\y) {};
\node[nod] (e) at (\s*.31+\x,-\s*.95+\y) {};
\draw[line width=1pt] (a) -- (b) -- (c) -- (d);
\node (o) at (\x,\y) {5};
\def\x{0}
\def\y{0}
\def\s{1}
\node[nod] (a) at (\s*1+\x,0+\y) {};
\node[nod] (b) at (\s*.31+\x,\s*.95+\y) {};
\node[nod] (c) at (-\s*.81+\x,\s*.59+\y) {};
\node[nod] (d) at (-\s*.81+\x,-\s*.59+\y) {};
\node[nod] (e) at (\s*.31+\x,-\s*.95+\y) {};
\draw[line width=1pt] (d) -- (c) -- (b);
\draw[line width=1pt] (a) -- (e);
\node (o) at (\x,0) {2};
\def\x{3}
\def\y{0}
\def\s{1}
\node[nod] (a) at (\s*1+\x,0+\y) {};
\node[nod] (b) at (\s*.31+\x,\s*.95+\y) {};
\node[nod] (c) at (-\s*.81+\x,\s*.59+\y) {};
\node[nod] (d) at (-\s*.81+\x,-\s*.59+\y) {};
\node[nod] (e) at (\s*.31+\x,-\s*.95+\y) {};
\draw[line width=1pt] (b) -- (c) -- (d);
\def\s{.9}
\coordinate (a) at (\s*1+\x,0+\y);
\coordinate (b) at (\s*.31+\x,\s*.95+\y);
\coordinate (c) at (-\s*.81+\x,\s*.59+\y);
\coordinate (d) at (-\s*.81+\x,-\s*.59+\y);
\coordinate (e) at (\s*.31+\x,-\s*.95+\y);
\draw[line width=.8pt] (d) -- (e) -- (a) -- (d);
\node (o) at (\x-.1,.1) {8};
\def\x{6}
\def\y{0}
\def\s{1}
\node[nod] (a) at (\s*1+\x,0+\y) {};
\node[nod] (b) at (\s*.31+\x,\s*.95+\y) {};
\node[nod] (c) at (-\s*.81+\x,\s*.59+\y) {};
\node[nod] (d) at (-\s*.81+\x,-\s*.59+\y) {};
\node[nod] (e) at (\s*.31+\x,-\s*.95+\y) {};
\draw[line width=1pt] (d) -- (c) -- (b);
\def\s{.9}
\coordinate (a) at (\s*1+\x,0+\y);
\coordinate (b) at (\s*.31+\x,\s*.95+\y);
\coordinate (c) at (-\s*.81+\x,\s*.59+\y);
\coordinate (d) at (-\s*.81+\x,-\s*.59+\y);
\coordinate (e) at (\s*.31+\x,-\s*.95+\y);
\draw[line width=.8pt] (c) -- (e) -- (a) -- (c);
\node (o) at (\x+.1,-.1) {3};
\def\x{9}
\def\y{0}
\def\s{1}
\node[nod] (a) at (\s*1+\x,0+\y) {};
\node[nod] (b) at (\s*.31+\x,\s*.95+\y) {};
\node[nod] (c) at (-\s*.81+\x,\s*.59+\y) {};
\node[nod] (d) at (-\s*.81+\x,-\s*.59+\y) {};
\node[nod] (e) at (\s*.31+\x,-\s*.95+\y) {};
\draw[line width=1pt] (d) -- (c);
\draw[line width=1pt] (a) -- (b);
\def\s{.9}
\coordinate (a) at (\s*1+\x,0+\y);
\coordinate (b) at (\s*.31+\x,\s*.95+\y);
\coordinate (c) at (-\s*.81+\x,\s*.59+\y);
\coordinate (d) at (-\s*.81+\x,-\s*.59+\y);
\coordinate (e) at (\s*.31+\x,-\s*.95+\y);
\draw[line width=.8pt] (d) -- (e) -- (a) -- (d);
\node (o) at (\x-.1,.1) {5};
\end{tikzpicture}
\caption{\label{h5k3} Nonequivalent admissible hypergraphs on 5 vertices with 3 hyperedges.}
\end{figure}

Associated LV-matrices are:
\[
\begin{pmatrix}
a_{1} & a_{8} & a_{9} & a_{6} & a_{11} 
\\
 a_{2} & a_{3} & a_{9} & a_{6} & a_{11} 
\\
 a_{4} & a_{8} & a_{5} & a_{6} & a_{11} 
\\
 a_{7} & a_{8} & a_{9} & a_{10} & a_{11} 
\\
 a_{12} & a_{13} & a_{14} & a_{15} & a_{16}   
\end{pmatrix},\qquad 
\begin{pmatrix}
 a_{1} & a_{5} & a_{3} & a_{7} & a_{8} 
\\
 a_{4} & a_{2} & a_{3} & a_{7} & a_{8} 
\\
 a_{4} & a_{5} & a_{6} & a_{7} & a_{8} 
\\
 a_{9} & a_{10} & a_{11} & a_{12} & a_{13} 
\\
 a_{14} & a_{15} & a_{16} & a_{17} & a_{18}   
\end{pmatrix}
\]
(which only has two independent integrals),
\[
\begin{pmatrix}
a_{1} & a_{4} & a_{9} & a_{6} & a_{11} 
\\
 a_{7} & a_{2} & a_{9} & a_{6} & a_{11} 
\\
 a_{3} & a_{4} & a_{5} & a_{6} & a_{11} 
\\
 a_{7} & a_{8} & a_{9} & a_{10} & a_{11} 
\\
 a_{12} & a_{13} & a_{14} & a_{15} & a_{16}    
\end{pmatrix},\qquad
\begin{pmatrix}
  a_{1} & a_{6} & a_{4} & a_{8} & a_{9} 
\\
 a_{2} & a_{3} & a_{4} & a_{8} & a_{9} 
\\
 a_{5} & a_{6} & a_{7} & a_{8} & a_{9} 
\\
 a_{12} & a_{13} & a_{14} & a_{10} & a_{11} 
\\
 a_{12} & a_{13} & a_{14} & a_{15} & a_{16} 
\end{pmatrix},
\]
\[
\begin{pmatrix}
  a_{1} & a_{5} & a_{12} & a_{3} & a_{7} 
\\
 a_{10} & a_{2} & a_{12} & a_{3} & a_{7} 
\\
 a_{4} & a_{5} & a_{6} & a_{3} & a_{7} 
\\
 a_{10} & a_{2}-\frac{\left(a_{3}-a_{8}\right) \left(a_{9}-a_{14}\right) \left(-a_{11}+a_{2}\right)}{\left(a_{7}-a_{14}\right) \left(a_{8}-a_{13}\right)} & a_{12} & a_{8} & a_{9} 
\\
 a_{10} & a_{11} & a_{12} & a_{13} & a_{14}  
\end{pmatrix},
\]
\[
\begin{pmatrix}
  a_{1} & a_{11} & a_{12} & a_{2} & a_{3} 
\\
 a_{4} & a_{5} & a_{12} & a_{2} & a_{3} 
\\
 a_{6} & a_{11} & a_{7} & a_{2} & a_{3} 
\\
 a_{1}-\frac{\left(a_{2}-a_{8}\right) \left(a_{9}-a_{14}\right) \left(-a_{10}+a_{1}\right)}{\left(a_{3}-a_{14}\right) \left(a_{8}-a_{13}\right)} & a_{11} & a_{12} & a_{8} & a_{9} 
\\
 a_{10} & a_{11} & a_{12} & a_{13} & a_{14}  
\end{pmatrix},
\]
\[
\begin{pmatrix}
  a_{1} & a_{11} & a_{2} & a_{13} & a_{3} 
\\
 a_{4} & a_{5} & a_{2} & a_{13} & a_{3} 
\\
 a_{10} & a_{11} & a_{6} & a_{13} & a_{7} 
\\
 a_{10} & a_{11} & a_{8} & a_{9} & a_{7} 
\\
 a_{1}-\frac{\left(a_{3}-a_{14}\right) \left(-a_{12}+a_{6}\right) \left(-a_{10}+a_{1}\right)}{\left(-a_{6}+a_{2}\right) \left(a_{7}-a_{14}\right)} & a_{11} & a_{12} & a_{13} & a_{14}  
\end{pmatrix}.
\]
\item Representatives of the size 4 admissible hypergraph LV-equivalence classes are depicted in Fig. \ref{h5k4}.
\begin{figure}[H]
\centering
\begin{tikzpicture}[scale=.8]
\tikzstyle{nod}= [circle, inner sep=0pt, fill=white, minimum size=2pt, draw]		
\def\x{0}
\def\y{3}
\def\s{1}
\node[nod] (a) at (\s*1+\x,0+\y) {};
\node[nod] (b) at (\s*.31+\x,\s*.95+\y) {};
\node[nod] (c) at (-\s*.81+\x,\s*.59+\y) {};
\node[nod] (d) at (-\s*.81+\x,-\s*.59+\y) {};
\node[nod] (e) at (\s*.31+\x,-\s*.95+\y) {};
\draw[line width=1pt] (c) -- (d) -- (b);
\draw[line width=1pt] (e) -- (d) -- (a);
\node (o) at (\x+.1,\y.1) {2};
\def\x{3}
\def\y{3}
\def\s{1}
\node[nod] (a) at (\s*1+\x,0+\y) {};
\node[nod] (b) at (\s*.31+\x,\s*.95+\y) {};
\node[nod] (c) at (-\s*.81+\x,\s*.59+\y) {};
\node[nod] (d) at (-\s*.81+\x,-\s*.59+\y) {};
\node[nod] (e) at (\s*.31+\x,-\s*.95+\y) {};
\draw[line width=1pt] (c) -- (d) -- (b);
\draw[line width=1pt] (d) -- (e) -- (a);
\node (o) at (\x+.1,\y-.1) {8};
\def\x{6}
\def\y{3}
\def\s{1}
\node[nod] (a) at (\s*1+\x,0+\y) {};
\node[nod] (b) at (\s*.31+\x,\s*.95+\y) {};
\node[nod] (c) at (-\s*.81+\x,\s*.59+\y) {};
\node[nod] (d) at (-\s*.81+\x,-\s*.59+\y) {};
\node[nod] (e) at (\s*.31+\x,-\s*.95+\y) {};
\draw[line width=1pt] (a) -- (b) -- (c) -- (d) -- (e);
\node (o) at (\x,\y) {12};
\def\x{0}
\def\y{0}
\def\s{1}
\node[nod] (a) at (\s*1+\x,0+\y) {};
\node[nod] (b) at (\s*.31+\x,\s*.95+\y) {};
\node[nod] (c) at (-\s*.81+\x,\s*.59+\y) {};
\node[nod] (d) at (-\s*.81+\x,-\s*.59+\y) {};
\node[nod] (e) at (\s*.31+\x,-\s*.95+\y) {};
\draw[line width=1pt] (d) -- (c) -- (b) -- (d);
\draw[line width=1pt] (a) -- (e);
\node (o) at (\x+.1,-.1) {2};
\def\x{3}
\def\y{0}
\def\s{1}
\node[nod] (a) at (\s*1+\x,0+\y) {};
\node[nod] (b) at (\s*.31+\x,\s*.95+\y) {};
\node[nod] (c) at (-\s*.81+\x,\s*.59+\y) {};
\node[nod] (d) at (-\s*.81+\x,-\s*.59+\y) {};
\node[nod] (e) at (\s*.31+\x,-\s*.95+\y) {};
\draw[line width=1pt] (b) -- (c) -- (d) -- (b) -- (a);
\node (o) at (\x+.1,-.1) {9};
\def\x{6}
\def\y{0}
\def\s{1}
\node[nod] (a) at (\s*1+\x,0+\y) {};
\node[nod] (b) at (\s*.31+\x,\s*.95+\y) {};
\node[nod] (c) at (-\s*.81+\x,\s*.59+\y) {};
\node[nod] (d) at (-\s*.81+\x,-\s*.59+\y) {};
\node[nod] (e) at (\s*.31+\x,-\s*.95+\y) {};
\draw[line width=1pt] (d) -- (b) -- (c) -- (d);
\def\s{.9}
\coordinate (a) at (\s*1+\x,0+\y);
\coordinate (b) at (\s*.31+\x,\s*.95+\y);
\coordinate (c) at (-\s*.81+\x,\s*.59+\y);
\coordinate (d) at (-\s*.81+\x,-\s*.59+\y);
\coordinate (e) at (\s*.31+\x,-\s*.95+\y);
\draw[line width=.8pt] (d) -- (e) -- (a) -- (d);
\node (o) at (\x+.1,+.1) {9};
\def\x{9}
\def\y{0}
\def\s{1}
\node[nod] (a) at (\s*1+\x,0+\y) {};
\node[nod] (b) at (\s*.31+\x,\s*.95+\y) {};
\node[nod] (c) at (-\s*.81+\x,\s*.59+\y) {};
\node[nod] (d) at (-\s*.81+\x,-\s*.59+\y) {};
\node[nod] (e) at (\s*.31+\x,-\s*.95+\y) {};
\draw[line width=1pt] (a) -- (b) -- (c) -- (d);
\def\s{.9}
\coordinate (a) at (\s*1+\x,0+\y);
\coordinate (b) at (\s*.31+\x,\s*.95+\y);
\coordinate (c) at (-\s*.81+\x,\s*.59+\y);
\coordinate (d) at (-\s*.81+\x,-\s*.59+\y);
\coordinate (e) at (\s*.31+\x,-\s*.95+\y);
\draw[line width=.8pt] (b) -- (c) -- (e) -- (b);
\node (o) at (\x,+.1) {3};
\end{tikzpicture}
\caption{\label{h5k4} Nonequivalent admissible hypergraphs on 5 vertices with 4 hyperedges.}
\end{figure}
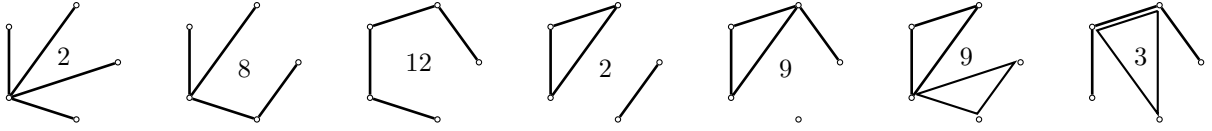

Associated LV-matrices are:
\[
\begin{pmatrix}
  a_{1} & a_{10} & a_{11} & a_{12} & a_{8} 
\\
 a_{2} & a_{3} & a_{11} & a_{12} & a_{8} 
\\
 a_{4} & a_{10} & a_{5} & a_{12} & a_{8} 
\\
 a_{6} & a_{10} & a_{11} & a_{7} & a_{8} 
\\
 a_{9} & a_{10} & a_{11} & a_{12} & a_{13} 
\end{pmatrix},\qquad
\begin{pmatrix}
  a_{1} & a_{6} & a_{11} & a_{12} & a_{8} 
\\
 a_{9} & a_{2} & a_{11} & a_{12} & a_{8} 
\\
 a_{3} & a_{6} & a_{4} & a_{12} & a_{8} 
\\
 a_{5} & a_{6} & a_{11} & a_{7} & a_{8} 
\\
 a_{9} & a_{10} & a_{11} & a_{12} & a_{13}  
\end{pmatrix},
\]
\begin{equation} \label{e1}
\begin{pmatrix}
  a_{1} & a_{10} & a_{6} & a_{12} & a_{8} 
\\
 a_{4} & a_{2} & a_{6} & a_{12} & a_{8} 
\\
 a_{9} & a_{10} & a_{3} & a_{12} & a_{8} 
\\
 a_{4} & a_{5} & a_{6} & a_{7} & a_{8} 
\\
 a_{9} & a_{10} & a_{11} & a_{12} & a_{13}  
\end{pmatrix}, \qquad
\begin{pmatrix}
  a_{1} & a_{5} & a_{3} & a_{7} & a_{8} 
\\
 a_{4} & a_{2} & a_{3} & a_{7} & a_{8} 
\\
 a_{4} & a_{5} & a_{6} & a_{7} & a_{8} 
\\
 a_{11} & a_{12} & a_{13} & a_{9} & a_{10} 
\\
 a_{11} & a_{12} & a_{13} & a_{14} & a_{15}  
\end{pmatrix},
\end{equation}
\begin{equation} \label{e2}
\begin{pmatrix}
  a_{1} & a_{7} & a_{8} & a_{5} & a_{10} 
\\
 a_{3} & a_{2} & a_{8} & a_{5} & a_{10} 
\\
 a_{3} & a_{7} & a_{4} & a_{5} & a_{10} 
\\
 a_{6} & a_{7} & a_{8} & a_{9} & a_{10} 
\\
 a_{11} & a_{12} & a_{13} & a_{14} & a_{15} 
\end{pmatrix}, \qquad
\begin{pmatrix}
  a_{1} & a_{10} & a_{11} & a_{2} & a_{3} 
\\
 a_{5} & a_{4} & a_{11} & a_{2} & a_{3} 
\\
 a_{5} & a_{10} & a_{6} & a_{2} & a_{3} 
\\
 a_{1}-\frac{\left(-a_{7}+a_{2}\right) \left(-a_{9}+a_{1}\right) \left(a_{8}-a_{13}\right)}{\left(a_{3}-a_{13}\right) \left(-a_{12}+a_{7}\right)} & a_{10} & a_{11} & a_{7} & a_{8} 
\\
 a_{9} & a_{10} & a_{11} & a_{12} & a_{13}  
\end{pmatrix}
\end{equation}
\begin{equation} \label{e3}
\begin{pmatrix}
  a_{1} & a_{2} & a_{11} & a_{12} & a_{3} 
\\
 a_{4} & a_{5} & a_{11} & a_{12} & a_{3} 
\\
 a_{6} & a_{2} & a_{7} & a_{12} & a_{3} 
\\
 a_{4} & a_{8} & a_{11} & a_{9} & a_{3} 
\\
 a_{1}-\frac{\left(a_{1}-a_{4}\right) \left(a_{5}-a_{10}\right)}{-a_{5}+a_{2}} & a_{10} & a_{11} & a_{12} & a_{13}  
\end{pmatrix}.
\end{equation}
We note that the LV-system related to the second matrix of \eqref{e1} and the LV-systems related to the matrices \eqref{e2} have 3 independent integrals only. The LV-system related to \eqref{e3} is a new superintegrable LV-system that is not equivalent to a tree-system. It has the same number of parameters as the 5-component tree-systems. The two nonisomorphic hypergraphs that are LV-equivalent are depicted in Fig. \ref{tnh}.

\begin{figure}[H]
\centering
\begin{tikzpicture}[scale=.8]
\tikzstyle{nod}= [circle, inner sep=0pt, fill=white, minimum size=2pt, draw]		
\def\x{0}
\def\y{0}
\def\s{1}
\node[nod] (a) at (\s*1+\x,0+\y) {};
\node[nod] (b) at (\s*.31+\x,\s*.95+\y) {};
\node[nod] (c) at (-\s*.81+\x,\s*.59+\y) {};
\node[nod] (d) at (-\s*.81+\x,-\s*.59+\y) {};
\node[nod] (e) at (\s*.31+\x,-\s*.95+\y) {};
\draw[line width=1pt] (e) -- (d);
\draw[line width=1pt] (b) -- (c);
\def\s{.9}
\coordinate (a) at (\s*1+\x,0+\y);
\coordinate (b) at (\s*.31+\x,\s*.95+\y);
\coordinate (c) at (-\s*.81+\x,\s*.59+\y);
\coordinate (d) at (-\s*.81+\x,-\s*.59+\y);
\coordinate (e) at (\s*.31+\x,-\s*.95+\y);
\draw[line width=.8pt] (d) -- (a) -- (b) -- (c) -- (d);
\def\s{.8}
\coordinate (a) at (\s*1+\x,0+\y);
\coordinate (b) at (\s*.31+\x,\s*.95+\y);
\coordinate (c) at (-\s*.81+\x,\s*.59+\y);
\coordinate (d) at (-\s*.81+\x,-\s*.59+\y);
\coordinate (e) at (\s*.31+\x,-\s*.95+\y);
\draw[line width=.8pt] (d) -- (b) -- (c) -- (d);
\def\x{3}
\def\y{0}
\def\s{1}
\node[nod] (a) at (\s*1+\x,0+\y) {};
\node[nod] (b) at (\s*.31+\x,\s*.95+\y) {};
\node[nod] (c) at (-\s*.81+\x,\s*.59+\y) {};
\node[nod] (d) at (-\s*.81+\x,-\s*.59+\y) {};
\node[nod] (e) at (\s*.31+\x,-\s*.95+\y) {};
\draw[line width=1pt] (c) -- (d);
\draw[line width=1pt] (b) -- (a);
\def\s{.9}
\coordinate (a) at (\s*1+\x,0+\y);
\coordinate (b) at (\s*.31+\x,\s*.95+\y);
\coordinate (c) at (-\s*.81+\x,\s*.59+\y);
\coordinate (d) at (-\s*.81+\x,-\s*.59+\y);
\coordinate (e) at (\s*.31+\x,-\s*.95+\y);
\draw[line width=.8pt] (d) -- (e) -- (a) -- (b) -- (c) -- (d);
\def\s{.8}
\coordinate (a) at (\s*1+\x,0+\y);
\coordinate (b) at (\s*.31+\x,\s*.95+\y);
\coordinate (c) at (-\s*.81+\x,\s*.59+\y);
\coordinate (d) at (-\s*.81+\x,-\s*.59+\y);
\coordinate (e) at (\s*.31+\x,-\s*.95+\y);
\draw[line width=.8pt] (d) -- (a) -- (b) -- (c) -- (d);
\end{tikzpicture}
\caption{\label{tnh} The two hypergraphs that are LV-equivalent to the last hypergraph in Fig. \ref{h5k4}.}
\end{figure}
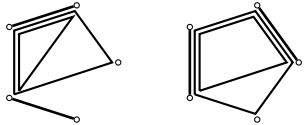

\item Representatives of the size 5 admissible hypergraph LV-equivalence classes are depicted in Fig. \ref{h5k5}.

\begin{figure}[H]
\centering
\begin{tikzpicture}[scale=.8]
\tikzstyle{nod}= [circle, inner sep=0pt, fill=white, minimum size=2pt, draw]		
\def\x{0}
\def\y{0}
\def\s{1}
\node[nod] (a) at (\s*1+\x,0+\y) {};
\node[nod] (b) at (\s*.31+\x,\s*.95+\y) {};
\node[nod] (c) at (-\s*.81+\x,\s*.59+\y) {};
\node[nod] (d) at (-\s*.81+\x,-\s*.59+\y) {};
\node[nod] (e) at (\s*.31+\x,-\s*.95+\y) {};
\draw[line width=1pt] (c) -- (d) -- (b) -- (c);
\draw[line width=1pt] (e) -- (d) -- (a);
\node (o) at (\x+.1,+.1) {9};
\def\x{3}
\def\y{0}
\def\s{1}
\node[nod] (a) at (\s*1+\x,0+\y) {};
\node[nod] (b) at (\s*.31+\x,\s*.95+\y) {};
\node[nod] (c) at (-\s*.81+\x,\s*.59+\y) {};
\node[nod] (d) at (-\s*.81+\x,-\s*.59+\y) {};
\node[nod] (e) at (\s*.31+\x,-\s*.95+\y) {};
\draw[line width=1pt] (b) -- (c) -- (d) -- (b) -- (a);
\draw[line width=1pt] (d) -- (e);
\node (o) at (\x+.1,-.1) {19};
\def\x{6}
\def\y{0}
\def\s{1}
\node[nod] (a) at (\s*1+\x,0+\y) {};
\node[nod] (b) at (\s*.31+\x,\s*.95+\y) {};
\node[nod] (c) at (-\s*.81+\x,\s*.59+\y) {};
\node[nod] (d) at (-\s*.81+\x,-\s*.59+\y) {};
\node[nod] (e) at (\s*.31+\x,-\s*.95+\y) {};
\draw[line width=1pt] (c) -- (d) -- (b) -- (c);
\draw[line width=1pt] (d) -- (e) -- (a);
\node (o) at (\x+.1,-.1) {24};
\def\x{9}
\def\y{0}
\def\s{1}
\node[nod] (a) at (\s*1+\x,0+\y) {};
\node[nod] (b) at (\s*.31+\x,\s*.95+\y) {};
\node[nod] (c) at (-\s*.81+\x,\s*.59+\y) {};
\node[nod] (d) at (-\s*.81+\x,-\s*.59+\y) {};
\node[nod] (e) at (\s*.31+\x,-\s*.95+\y) {};
\draw[line width=1pt] (a) -- (e) -- (d) -- (b) -- (e);
\def\s{.9}
\coordinate (a) at (\s*1+\x,0+\y);
\coordinate (b) at (\s*.17+\x,\s*.92+\y);
\coordinate (c) at (-\s*.81+\x,\s*.59+\y);
\coordinate (d) at (-\s*.81+\x,-\s*.42+\y);
\coordinate (e) at (\s*.31+\x,-\s*.95+\y);
\draw[line width=.8pt] (d) -- (c) -- (b) -- (d);
\node (o) at (\x,-.1) {8};
\end{tikzpicture}
\caption{\label{h5k5} Nonequivalent admissible hypergraphs on 5 vertices with 5 hyperedges.}
\end{figure}
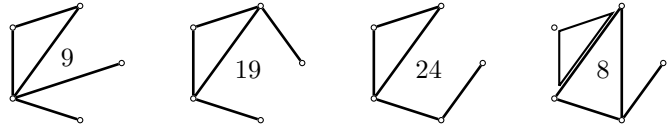

As the LV-systems are subsystems of the LV-systems with 4 DPS, we do not give their matrices explicitly. In each case, the number of parameters is 12.

\item Representatives of size 6,7 and 10 admissible hypergraph LV-equivalence classes are depicted in Fig. \ref{h5k6710}.
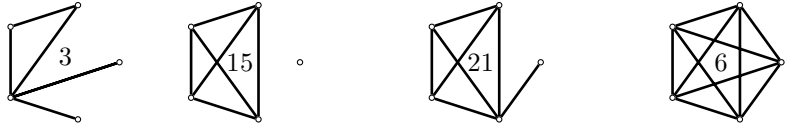
\begin{figure}[H]
\centering
\begin{tikzpicture}[scale=.8]
\tikzstyle{nod}= [circle, inner sep=0pt, fill=white, minimum size=2pt, draw]		
\def\x{0}
\def\y{0}
\def\s{1}
\node[nod] (a) at (\s*1+\x,0+\y) {};
\node[nod] (b) at (\s*.31+\x,\s*.95+\y) {};
\node[nod] (c) at (-\s*.81+\x,\s*.59+\y) {};
\node[nod] (d) at (-\s*.81+\x,-\s*.59+\y) {};
\node[nod] (e) at (\s*.31+\x,-\s*.95+\y) {};
\draw[line width=1pt] (c) -- (d) -- (b) -- (c);
\draw[line width=1pt] (e) -- (d) -- (a) -- (d);
\node (o) at (\x+.1,+.1) {3};
\def\x{3}
\def\y{0}
\def\s{1}
\node[nod] (a) at (\s*1+\x,0+\y) {};
\node[nod] (b) at (\s*.31+\x,\s*.95+\y) {};
\node[nod] (c) at (-\s*.81+\x,\s*.59+\y) {};
\node[nod] (d) at (-\s*.81+\x,-\s*.59+\y) {};
\node[nod] (e) at (\s*.31+\x,-\s*.95+\y) {};
\draw[line width=1pt] (b) -- (c) -- (d) -- (e) -- (c);
\draw[line width=1pt] (d) -- (b) -- (e);
\node (o) at (\x,0) {15};
\def\x{7}
\def\y{0}
\def\s{1}
\node[nod] (a) at (\s*1+\x,0+\y) {};
\node[nod] (b) at (\s*.31+\x,\s*.95+\y) {};
\node[nod] (c) at (-\s*.81+\x,\s*.59+\y) {};
\node[nod] (d) at (-\s*.81+\x,-\s*.59+\y) {};
\node[nod] (e) at (\s*.31+\x,-\s*.95+\y) {};
\draw[line width=1pt] (b) -- (c) -- (d) -- (e) -- (c);
\draw[line width=1pt] (d) -- (b) -- (e) -- (a);
\node (o) at (\x,0) {21};
\def\x{11}
\def\y{0}
\def\s{1}
\node[nod] (a) at (\s*1+\x,0+\y) {};
\node[nod] (b) at (\s*.31+\x,\s*.95+\y) {};
\node[nod] (c) at (-\s*.81+\x,\s*.59+\y) {};
\node[nod] (d) at (-\s*.81+\x,-\s*.59+\y) {};
\node[nod] (e) at (\s*.31+\x,-\s*.95+\y) {};
\draw[line width=1pt] (a) -- (b) -- (c) -- (d) -- (e) -- (a) -- (c) -- (e) -- (b) -- (d) -- (a);
\node (o) at (\x,0) {6};
\end{tikzpicture}
\caption{\label{h5k6710} Nonequivalent admissible hypergraphs on 5 vertices with 6,7 or 10 hyperedges.}
\end{figure}

The number of parameters in these LV-systems is respectively 14, 11, 11, and 10. Any LV-system with $k>7$ DPs also has $7$ DPs. There are 9 nonisomorphic hypergraphs of size 8 which have the representative admissible hypergraph of size 7 as subgraph. None of these are admissible, but two of them give rise to a hypergraph of size 10 (cf. Examples \ref{exam1} and \ref{exam2}). There are 6 nonisomorphic hypergraphs of size 10 that are LV-equivalent to the complete graph of order 5. 
\end{enumerate}

\section{Concluding remarks}
We have classified homogeneous $n$-component LV-systems with $k$ additional linear Darboux polynomials, up to linear transformations, for $n\leq5$. The number of them is given in Table \ref{cls}.
\begin{table}[h]
\begin{center}
\begin{tabular}{c|cccccccccc}
$n \setminus k$ & 1 & 2 & 3 & 4 & 5 & 6 & 7 & 8 & 9 & 10 \\
\hline
2 & 1 & & & & & & & & & \\
3 & 2 & 1 & 1 & 0 & & & & & & \\
4 & 3 & 3 & 3 & 1 & 0 & 1 & 0 & 0 & 0 & 0 \\
5 & 4 & 6 & 7 & 7 & 4 & 2 & 1 & 0 & 0 & 1 \\
\end{tabular}
\vspace{2mm}

\parbox{12.5cm}{\caption{\label{cls} The number of $n$-component LV-systems with $k$ additional linear DPs.}}
\end{center}
\end{table}
\vspace{-8pt} 

Each DP gives rise to an integral, however, they are not necessarily functionally independent. The maximum number of independent integrals for an $n$-component system is $n-1$. For an admissible hypergraph of size $k$, which contains a subgraph of order $o$ and size $l>o-1$ the number of independent integrals for the corresponding LV-system is at most $k-l+o-1$. The number of LV-systems with functionally independent integrals of the form \eqref{Ik} is given in Table \ref{cli}. (Note that LV-systems with more linear DPs than independent integrals are subsystems of a larger system, i.e. with more pararameters. These are not counted.)

\begin{table}[h]
\begin{center}
\begin{tabular}{c|cccccccccc}
$n \setminus k$ & 1 & 2 & 3 & 4 \\
\hline
2 & 1 & & & \\
3 & 2 & 1 & &  \\
4 & 3 & 3 & 2 & \\
5 & 4 & 6 & 6 & 4 \\
\end{tabular}
\vspace{2mm}

\parbox{12.5cm}{\caption{\label{cli} The number of $n$-component LV-systems with $k$ functionally independent integrals of the form \eqref{Ik}.}}
\end{center}
\end{table}
\vspace{-8pt} 

The homogeneous systems we have classified, of the form \eqref{LVsys} with $\b={\bf 0}$, can be generalised to nonhomogeneous systems as follows. Suppose the hypergraph $H$ of the homogeneous LV-system is a hyperforest. By this we mean that the set $\N_n$ is a disjoint union of sets $\N_n=\cup_{j=1}^r M_j$ such that for each hyperedge $h\in H$ we have $h\in M_j$ for some $j$. Then, for arbitrary constants $c_j$, $j=1,\ldots,r$, the nonhomogeneous system \eqref{LVsys} with $b_i=c_j$ for all $i\in M_j$ has the same DPs.

Of special interest are the $n$-component LV-systems which have $n-1$ functionally independent integrals, the so-called superintegrable ones. A large class of such systems are associated to tree-systems, cf. \cite{QTMK,trees}. For $n\leq 4$ all superintegrable hypergraph-systems are tree-systems. For $n=5$ this is not true. The LV-system with matrix \eqref{e3} is a hypergraph-system which is not a tree-system. For $n>5$, we expect there exist more superintegrable hypergraph-systems that are not equivalent to tree-systems.

The reader should be able, for each of the matrices we have provided, to figure out what the labeling is on the corresponding hypergraph, by inspecting which of the DP-conditions are satisfied, cf. Lemma \ref{lem} and Remark \ref{oc}. For the matrix \eqref{e3}, which corresponds to the right-most hypergraph in Fig. \ref{h5k4}, it is easy to see
that the 2-edges are $\{1,2\},\{1,3\},\{2,4\}$ and the 3-edge is $\{1,2,5\}$. The DPs are
\begin{align*}
&\left(a_{4}-a_{1}\right) x_{1}+\left(a_{5}-a_{2}\right) x_{2},\quad
\left(a_{6}-a_{1}\right) x_{1}+\left(a_{7}-a_{11}\right) x_{3},\quad
\left(a_{8}-a_{5}\right) x_{2}+\left(a_{9}-a_{12}\right) x_{4},\\
&\left(a_{1}-a_{4}\right)\left(a_{5}-a_{10}\right)x_{1}
+\left(a_{2}-a_{5}\right)\left(a_{5}-a_{10}\right)x_{2}
+\left(a_{2}-a_{5}\right)\left(a_{3}-a_{13}\right)x_{5},
\end{align*}
and 4 functionally independent integrals can be obtained from the DPs by using equation \eqref{Ik}.

We have determined the sets of hypergraphs that are LV-equivalent to trees of order $n<9$. Their sizes are given in Table \ref{seh}, as well as in the appendix. It confirms the following statement for trees of order smaller than nine.
\begin{conjecture}
Nonisomorphic trees are not LV-equivalent.
\end{conjecture}

\begin{table}[H]
\begin{center}
\begin{tabular}{c|c|c}
$n$ & trees & LV-equivalent hypergraphs \\
\hline
2 & 1 & 1 \\
3 & 1 & 2 \\
4 & 2 & 2, 5 \\
5 & 3 & 2, 8, 12 \\
6 & 6 & 2, 8, 5, 18, 21, 30 \\
7 & 11 & 2, 8, 8, 18, 21, 30, 22, 79, 12, 55, 76 \\
8 & 23 & 2, 8, 8, 18, 21, 5, 30, 30, 32, 79, 21, 55, 18, 48, 60, 79, 126, 79, 112, 207, 30, 144, 195 
\end{tabular}
\vspace{2mm}

\parbox{12.5cm}{\caption{\label{seh} The number of trees of order $n<9$ and the numbers of LV-equivalent hypergraphs.}}
\end{center}
\end{table}

\appendix

\section{Trees of orders $6,7,8$ and the number of LV-equivalent hypergraphs}
In Figs. \ref{t6}, \ref{t7} and \ref{t8}, we provide for each tree of order 6,7 and 8 the number of LV-equivalent hypergraphs.

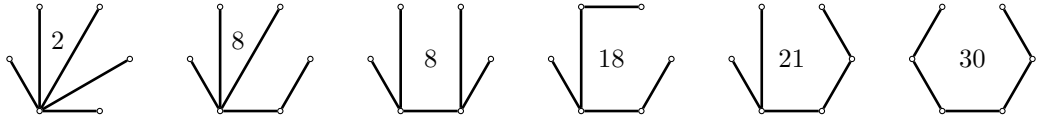
\begin{figure}[H]
\centering
\begin{tikzpicture}[scale=.8]
\tikzstyle{nod}= [circle, inner sep=0pt, fill=white, minimum size=2pt, draw]		
\def\x{0}
\def\y{0}
\node[nod] (a) at (1+\x,\y) {};
\node[nod] (b) at (.5+\x,.866+\y) {};
\node[nod] (c) at (-.5+\x,.866+\y) {};
\node[nod] (d) at (-1+\x,\y) {};
\node[nod] (e) at (-.5+\x,-.866+\y) {};
\node[nod] (f) at (.5+\x,-.866+\y) {};
\draw[line width=1pt] (d) -- (e) -- (b); \draw[line width=1pt] (c) -- (e) -- (a);
\draw[line width=1pt] (e) -- (f);
\node (o) at (\x-.2,\y+.3) {2};
\def\x{3}
\def\y{0}
\node[nod] (a) at (1+\x,\y) {};
\node[nod] (b) at (.5+\x,.866+\y) {};
\node[nod] (c) at (-.5+\x,.866+\y) {};
\node[nod] (d) at (-1+\x,\y) {};
\node[nod] (e) at (-.5+\x,-.866+\y) {};
\node[nod] (f) at (.5+\x,-.866+\y) {};
\draw[line width=1pt] (d) -- (e) -- (b); \draw[line width=1pt] (c) -- (e);
\draw[line width=1pt] (e) -- (f) -- (a);
\node (o) at (\x-.2,\y+.3) {8};
\def\x{6}
\def\y{0}
\node[nod] (a) at (1+\x,\y) {};
\node[nod] (b) at (.5+\x,.866+\y) {};
\node[nod] (c) at (-.5+\x,.866+\y) {};
\node[nod] (d) at (-1+\x,\y) {};
\node[nod] (e) at (-.5+\x,-.866+\y) {};
\node[nod] (f) at (.5+\x,-.866+\y) {};
\draw[line width=1pt] (d) -- (e) -- (c); \draw[line width=1pt] (f) -- (b);
\draw[line width=1pt] (e) -- (f) -- (a);
\node (o) at (\x,\y) {8};
\def\x{9}
\def\y{0}
\node[nod] (a) at (1+\x,\y) {};
\node[nod] (b) at (.5+\x,.866+\y) {};
\node[nod] (c) at (-.5+\x,.866+\y) {};
\node[nod] (d) at (-1+\x,\y) {};
\node[nod] (e) at (-.5+\x,-.866+\y) {};
\node[nod] (f) at (.5+\x,-.866+\y) {};
\draw[line width=1pt] (d) -- (e) -- (c) -- (b); 
\draw[line width=1pt] (e) -- (f) -- (a);
\node (o) at (\x,\y) {18};
\def\x{12}
\def\y{0}
\node[nod] (a) at (1+\x,\y) {};
\node[nod] (b) at (.5+\x,.866+\y) {};
\node[nod] (c) at (-.5+\x,.866+\y) {};
\node[nod] (d) at (-1+\x,\y) {};
\node[nod] (e) at (-.5+\x,-.866+\y) {};
\node[nod] (f) at (.5+\x,-.866+\y) {};
\draw[line width=1pt] (d) -- (e) -- (c) ; 
\draw[line width=1pt] (e) -- (f) -- (a)-- (b);
\node (o) at (\x,\y) {21};
\def\x{15}
\def\y{0}
\node[nod] (a) at (1+\x,\y) {};
\node[nod] (b) at (.5+\x,.866+\y) {};
\node[nod] (c) at (-.5+\x,.866+\y) {};
\node[nod] (d) at (-1+\x,\y) {};
\node[nod] (e) at (-.5+\x,-.866+\y) {};
\node[nod] (f) at (.5+\x,-.866+\y) {};
\draw[line width=1pt] (c) -- (d) -- (e) ; 
\draw[line width=1pt] (e) -- (f) -- (a)-- (b);
\node (o) at (\x,\y) {30};
\end{tikzpicture}
\caption{\label{t6} Trees of order 6 and the number of LV-equivalent nonisomorphic hypergraphs.}
\end{figure}

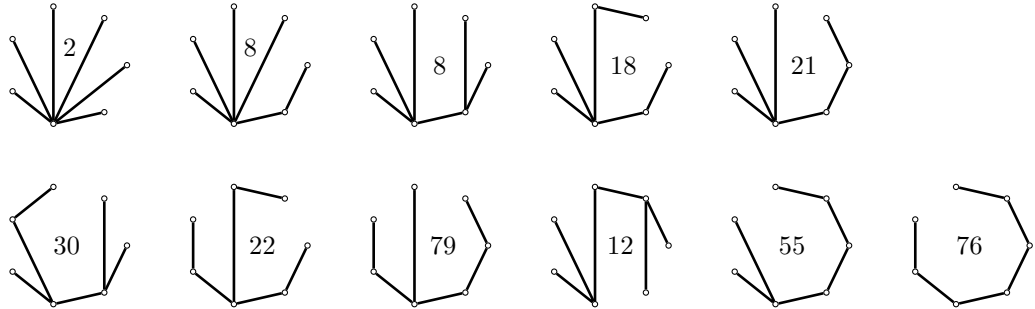
\begin{figure}[H]
\centering
\begin{tikzpicture}[scale=.8]
\tikzstyle{nod}= [circle, inner sep=0pt, fill=white, minimum size=2pt, draw]		
\def\x{0}
\def\y{3}
\node[nod] (a) at (1+\x,\y) {};
\node[nod] (b) at (.623+\x,.782+\y) {};
\node[nod] (c) at (-.223+\x,.975+\y) {};
\node[nod] (d) at (-.901+\x,.434+\y) {};
\node[nod] (e) at (-.901+\x,-.434+\y) {};
\node[nod] (f) at (-.223+\x,-.975+\y) {};
\node[nod] (g) at (.623+\x,-.782+\y) {};
\draw[line width=1pt] (d) -- (f) -- (b); \draw[line width=1pt] (c) -- (f) -- (a);
\draw[line width=1pt] (g) -- (f) -- (e);
\node (o) at (\x+.05,\y+.3) {2};
\def\x{3}
\def\y{3}
\node[nod] (a) at (1+\x,\y) {};
\node[nod] (b) at (.623+\x,.782+\y) {};
\node[nod] (c) at (-.223+\x,.975+\y) {};
\node[nod] (d) at (-.901+\x,.434+\y) {};
\node[nod] (e) at (-.901+\x,-.434+\y) {};
\node[nod] (f) at (-.223+\x,-.975+\y) {};
\node[nod] (g) at (.623+\x,-.782+\y) {};
\draw[line width=1pt] (d) -- (f) -- (b); \draw[line width=1pt] (c) -- (f);
\draw[line width=1pt] (e) -- (f) -- (g) -- (a);
\node (o) at (\x+.05,\y+.3) {8};
\def\x{6}
\def\y{3}
\node[nod] (a) at (1+\x,\y) {};
\node[nod] (b) at (.623+\x,.782+\y) {};
\node[nod] (c) at (-.223+\x,.975+\y) {};
\node[nod] (d) at (-.901+\x,.434+\y) {};
\node[nod] (e) at (-.901+\x,-.434+\y) {};
\node[nod] (f) at (-.223+\x,-.975+\y) {};
\node[nod] (g) at (.623+\x,-.782+\y) {};
\draw[line width=1pt] (d) -- (f) -- (c); \draw[line width=1pt] (g) -- (b);
\draw[line width=1pt] (e) -- (f) -- (g) -- (a);
\node (o) at (\x+.2,\y) {8};
\def\x{9}
\def\y{3}
\node[nod] (a) at (1+\x,\y) {};
\node[nod] (b) at (.623+\x,.782+\y) {};
\node[nod] (c) at (-.223+\x,.975+\y) {};
\node[nod] (d) at (-.901+\x,.434+\y) {};
\node[nod] (e) at (-.901+\x,-.434+\y) {};
\node[nod] (f) at (-.223+\x,-.975+\y) {};
\node[nod] (g) at (.623+\x,-.782+\y) {};
\draw[line width=1pt] (d) -- (f) -- (c) -- (b);
\draw[line width=1pt] (e) -- (f) -- (g) -- (a);
\node (o) at (\x+.25,\y) {18};
\def\x{12}
\def\y{3}
\node[nod] (a) at (1+\x,\y) {};
\node[nod] (b) at (.623+\x,.782+\y) {};
\node[nod] (c) at (-.223+\x,.975+\y) {};
\node[nod] (d) at (-.901+\x,.434+\y) {};
\node[nod] (e) at (-.901+\x,-.434+\y) {};
\node[nod] (f) at (-.223+\x,-.975+\y) {};
\node[nod] (g) at (.623+\x,-.782+\y) {};
\draw[line width=1pt] (d) -- (f) -- (c);
\draw[line width=1pt] (e) -- (f) -- (g) -- (a) -- (b);
\node (o) at (\x+.25,\y) {21};
\def\x{0}
\def\y{0}
\node[nod] (a) at (1+\x,\y) {};
\node[nod] (b) at (.623+\x,.782+\y) {};
\node[nod] (c) at (-.223+\x,.975+\y) {};
\node[nod] (d) at (-.901+\x,.434+\y) {};
\node[nod] (e) at (-.901+\x,-.434+\y) {};
\node[nod] (f) at (-.223+\x,-.975+\y) {};
\node[nod] (g) at (.623+\x,-.782+\y) {};
\draw[line width=1pt] (c) -- (d) -- (f) -- (e);
\draw[line width=1pt] (f) -- (g) -- (a); \draw[line width=1pt] (g) -- (b); 
\node (o) at (\x,\y) {30};
\def\x{3}
\def\y{0}
\node[nod] (a) at (1+\x,\y) {};
\node[nod] (b) at (.623+\x,.782+\y) {};
\node[nod] (c) at (-.223+\x,.975+\y) {};
\node[nod] (d) at (-.901+\x,.434+\y) {};
\node[nod] (e) at (-.901+\x,-.434+\y) {};
\node[nod] (f) at (-.223+\x,-.975+\y) {};
\node[nod] (g) at (.623+\x,-.782+\y) {};
\draw[line width=1pt] (f) -- (e) -- (d);
\draw[line width=1pt] (f) -- (c) -- (b); \draw[line width=1pt] (f) -- (g) -- (a); 
\node (o) at (\x+.25,\y) {22};
\def\x{6}
\def\y{0}
\node[nod] (a) at (1+\x,\y) {};
\node[nod] (b) at (.623+\x,.782+\y) {};
\node[nod] (c) at (-.223+\x,.975+\y) {};
\node[nod] (d) at (-.901+\x,.434+\y) {};
\node[nod] (e) at (-.901+\x,-.434+\y) {};
\node[nod] (f) at (-.223+\x,-.975+\y) {};
\node[nod] (g) at (.623+\x,-.782+\y) {};
\draw[line width=1pt] (f) -- (e) -- (d);
\draw[line width=1pt] (f) -- (c); \draw[line width=1pt] (f) -- (g) -- (a) -- (b); 
\node (o) at (\x+.25,\y) {79};
\def\x{9}
\def\y{0}
\node[nod] (a) at (1+\x,\y) {};
\node[nod] (b) at (.623+\x,.782+\y) {};
\node[nod] (c) at (-.223+\x,.975+\y) {};
\node[nod] (d) at (-.901+\x,.434+\y) {};
\node[nod] (e) at (-.901+\x,-.434+\y) {};
\node[nod] (f) at (-.223+\x,-.975+\y) {};
\node[nod] (g) at (.623+\x,-.782+\y) {};
\draw[line width=1pt] (e) -- (f) -- (d);
\draw[line width=1pt] (f) -- (c) -- (b) -- (a);
\draw[line width=1pt] (b) -- (g);
\node (o) at (\x+.2,\y) {12};
\def\x{12}
\def\y{0}
\node[nod] (a) at (1+\x,\y) {};
\node[nod] (b) at (.623+\x,.782+\y) {};
\node[nod] (c) at (-.223+\x,.975+\y) {};
\node[nod] (d) at (-.901+\x,.434+\y) {};
\node[nod] (e) at (-.901+\x,-.434+\y) {};
\node[nod] (f) at (-.223+\x,-.975+\y) {};
\node[nod] (g) at (.623+\x,-.782+\y) {};
\draw[line width=1pt] (e) -- (f) -- (d);
\draw[line width=1pt] (f) -- (g) -- (a) -- (b) -- (c);
\node (o) at (\x+.05,\y) {55};
\def\x{15}
\def\y{0}
\node[nod] (a) at (1+\x,\y) {};
\node[nod] (b) at (.623+\x,.782+\y) {};
\node[nod] (c) at (-.223+\x,.975+\y) {};
\node[nod] (d) at (-.901+\x,.434+\y) {};
\node[nod] (e) at (-.901+\x,-.434+\y) {};
\node[nod] (f) at (-.223+\x,-.975+\y) {};
\node[nod] (g) at (.623+\x,-.782+\y) {};
\draw[line width=1pt] (d) -- (e) -- (f) -- (g) -- (a) -- (b) -- (c);
\node (o) at (\x,\y) {76};
\end{tikzpicture}
\caption{\label{t7} Trees of order 7 and the number of LV-equivalent nonisomorphic hypergraphs.}
\end{figure}

\begin{figure}[H]
\centering
\begin{tikzpicture}[scale=.8]
\tikzstyle{nod}= [circle, inner sep=0pt, fill=white, minimum size=2pt, draw]
\def\x{0}
\def\y{9}
\node[nod] (a) at (1+\x,\y) {};
\node[nod] (b) at (.707+\x,.707+\y) {};
\node[nod] (c) at (0+\x,1+\y) {};
\node[nod] (d) at (-.707+\x,.707+\y) {};
\node[nod] (e) at (-1+\x,0+\y) {};
\node[nod] (f) at (-.707+\x,-.707+\y) {};
\node[nod] (g) at (\x,-1+\y) {};
\node[nod] (h) at (.707+\x,-.707+\y) {};
\draw[line width=1pt] (a) -- (f) -- (e); \draw[line width=1pt] (b) -- (f) -- (g);
\draw[line width=1pt] (c) -- (f) -- (h); \draw[line width=1pt] (d) -- (f);
\node (o) at (\x+.15,\y+.6) {2};
\def\x{3}
\def\y{9}
\node[nod] (a) at (1+\x,\y) {};
\node[nod] (b) at (.707+\x,.707+\y) {};
\node[nod] (c) at (0+\x,1+\y) {};
\node[nod] (d) at (-.707+\x,.707+\y) {};
\node[nod] (e) at (-1+\x,0+\y) {};
\node[nod] (f) at (-.707+\x,-.707+\y) {};
\node[nod] (g) at (\x,-1+\y) {};
\node[nod] (h) at (.707+\x,-.707+\y) {};
\draw[line width=1pt] (a) -- (f) -- (e); \draw[line width=1pt] (b) -- (f) -- (g)-- (h);
\draw[line width=1pt] (c) -- (f); \draw[line width=1pt] (d) -- (f);
\node (o) at (\x+.15,\y+.6) {8};
\def\x{6}
\def\y{9}
\node[nod] (a) at (1+\x,\y) {};
\node[nod] (b) at (.707+\x,.707+\y) {};
\node[nod] (c) at (0+\x,1+\y) {};
\node[nod] (d) at (-.707+\x,.707+\y) {};
\node[nod] (e) at (-1+\x,0+\y) {};
\node[nod] (f) at (-.707+\x,-.707+\y) {};
\node[nod] (g) at (\x,-1+\y) {};
\node[nod] (h) at (.707+\x,-.707+\y) {};
\draw[line width=1pt] (f) -- (e); 
\draw[line width=1pt] (g) -- (a); \draw[line width=1pt] (b) -- (f) -- (g)-- (h);
\draw[line width=1pt] (c) -- (f); \draw[line width=1pt] (d) -- (f);
\node (o) at (\x+.15,\y+.6) {8};
\def\x{9}
\def\y{9}
\node[nod] (a) at (1+\x,\y) {};
\node[nod] (b) at (.707+\x,.707+\y) {};
\node[nod] (c) at (0+\x,1+\y) {};
\node[nod] (d) at (-.707+\x,.707+\y) {};
\node[nod] (e) at (-1+\x,0+\y) {};
\node[nod] (f) at (-.707+\x,-.707+\y) {};
\node[nod] (g) at (\x,-1+\y) {};
\node[nod] (h) at (.707+\x,-.707+\y) {};
\draw[line width=1pt] (f) -- (e); 
\draw[line width=1pt] (a) -- (b) -- (f) -- (g)-- (h);
\draw[line width=1pt] (c) -- (f); \draw[line width=1pt] (d) -- (f);
\node (o) at (\x+.3,\y-.3) {18};
\def\x{12}
\def\y{9}
\node[nod] (a) at (1+\x,\y) {};
\node[nod] (b) at (.707+\x,.707+\y) {};
\node[nod] (c) at (0+\x,1+\y) {};
\node[nod] (d) at (-.707+\x,.707+\y) {};
\node[nod] (e) at (-1+\x,0+\y) {};
\node[nod] (f) at (-.707+\x,-.707+\y) {};
\node[nod] (g) at (\x,-1+\y) {};
\node[nod] (h) at (.707+\x,-.707+\y) {};
\draw[line width=1pt] (f) -- (e); 
\draw[line width=1pt] (b) -- (f) -- (g)-- (h) -- (a);
\draw[line width=1pt] (c) -- (f) --(d); 
\node (o) at (\x+.3,\y-.3) {21};
\def\x{0}
\def\y{6}
\node[nod] (a) at (1+\x,\y) {};
\node[nod] (b) at (.707+\x,.707+\y) {};
\node[nod] (c) at (0+\x,1+\y) {};
\node[nod] (d) at (-.707+\x,.707+\y) {};
\node[nod] (e) at (-1+\x,0+\y) {};
\node[nod] (f) at (-.707+\x,-.707+\y) {};
\node[nod] (g) at (\x,-1+\y) {};
\node[nod] (h) at (.707+\x,-.707+\y) {};
\draw[line width=1pt] (b) -- (g) -- (f) -- (e); 
\draw[line width=1pt] (d) -- (f) -- (c);
\draw[line width=1pt] (h) -- (g) --(a); 
\node (o) at (\x,\y)  {5};
\def\x{3}
\def\y{6}
\node[nod] (a) at (1+\x,\y) {};
\node[nod] (b) at (.707+\x,.707+\y) {};
\node[nod] (c) at (0+\x,1+\y) {};
\node[nod] (d) at (-.707+\x,.707+\y) {};
\node[nod] (e) at (-1+\x,0+\y) {};
\node[nod] (f) at (-.707+\x,-.707+\y) {};
\node[nod] (g) at (\x,-1+\y) {};
\node[nod] (h) at (.707+\x,-.707+\y) {};
\draw[line width=1pt] (g) -- (f) -- (e); 
\draw[line width=1pt] (d) -- (f) -- (c) -- (b);
\draw[line width=1pt] (h) -- (g) -- (a); 
\node (o) at (\x+.2,\y)   {30};
\def\x{6}
\def\y{6}
\node[nod] (a) at (1+\x,\y) {};
\node[nod] (b) at (.707+\x,.707+\y) {};
\node[nod] (c) at (0+\x,1+\y) {};
\node[nod] (d) at (-.707+\x,.707+\y) {};
\node[nod] (e) at (-1+\x,0+\y) {};
\node[nod] (f) at (-.707+\x,-.707+\y) {};
\node[nod] (g) at (\x,-1+\y) {};
\node[nod] (h) at (.707+\x,-.707+\y) {};
\draw[line width=1pt] (g) -- (f) -- (e); 
\draw[line width=1pt] (d) -- (f) -- (c);
\draw[line width=1pt] (h) -- (g) -- (a) -- (b); 
\node (o) at (\x+.2,\y)    {30};
\def\x{9}
\def\y{6}
\node[nod] (a) at (1+\x,\y) {};
\node[nod] (b) at (.707+\x,.707+\y) {};
\node[nod] (c) at (0+\x,1+\y) {};
\node[nod] (d) at (-.707+\x,.707+\y) {};
\node[nod] (e) at (-1+\x,0+\y) {};
\node[nod] (f) at (-.707+\x,-.707+\y) {};
\node[nod] (g) at (\x,-1+\y) {};
\node[nod] (h) at (.707+\x,-.707+\y) {};
\draw[line width=1pt] (g) -- (f) -- (e) -- (d);
\draw[line width=1pt] (a) -- (f); 
\draw[line width=1pt] (f) -- (c) -- (b);
\draw[line width=1pt] (h) -- (g); 
\node (o) at (\x+.2,\y+.2)   {32};
\def\x{12}
\def\y{6}
\node[nod] (a) at (1+\x,\y) {};
\node[nod] (b) at (.707+\x,.707+\y) {};
\node[nod] (c) at (0+\x,1+\y) {};
\node[nod] (d) at (-.707+\x,.707+\y) {};
\node[nod] (e) at (-1+\x,0+\y) {};
\node[nod] (f) at (-.707+\x,-.707+\y) {};
\node[nod] (g) at (\x,-1+\y) {};
\node[nod] (h) at (.707+\x,-.707+\y) {};
\draw[line width=1pt] (g) -- (f) -- (e); 
\draw[line width=1pt] (d) -- (f) -- (c) -- (b);
\draw[line width=1pt] (a) -- (h) -- (g); 
\node (o) at (\x+.2,\y)   {79};
\def\x{15}
\def\y{6}
\node[nod] (a) at (1+\x,\y) {};
\node[nod] (b) at (.707+\x,.707+\y) {};
\node[nod] (c) at (0+\x,1+\y) {};
\node[nod] (d) at (-.707+\x,.707+\y) {};
\node[nod] (e) at (-1+\x,0+\y) {};
\node[nod] (f) at (-.707+\x,-.707+\y) {};
\node[nod] (g) at (\x,-1+\y) {};
\node[nod] (h) at (.707+\x,-.707+\y) {};
\draw[line width=1pt] (g) -- (f) -- (e); 
\draw[line width=1pt] (d) -- (f) -- (c); \draw[line width=1pt] (h) -- (b);
\draw[line width=1pt] (a) -- (h) -- (g); 
\node (o) at (\x+.2,\y)   {21};
\def\x{0}
\def\y{3}
\node[nod] (a) at (1+\x,\y) {};
\node[nod] (b) at (.707+\x,.707+\y) {};
\node[nod] (c) at (0+\x,1+\y) {};
\node[nod] (d) at (-.707+\x,.707+\y) {};
\node[nod] (e) at (-1+\x,0+\y) {};
\node[nod] (f) at (-.707+\x,-.707+\y) {};
\node[nod] (g) at (\x,-1+\y) {};
\node[nod] (h) at (.707+\x,-.707+\y) {};
\draw[line width=1pt] (g) -- (f) -- (e); 
\draw[line width=1pt] (d) -- (f) -- (c); 
\draw[line width=1pt] (b) -- (a) -- (h) -- (g); 
\node (o) at (\x+.2,\y)  {55};
\def\x{3}
\def\y{3}
\node[nod] (a) at (1+\x,\y) {};
\node[nod] (b) at (.707+\x,.707+\y) {};
\node[nod] (c) at (0+\x,1+\y) {};
\node[nod] (d) at (-.707+\x,.707+\y) {};
\node[nod] (e) at (-1+\x,0+\y) {};
\node[nod] (f) at (-.707+\x,-.707+\y) {};
\node[nod] (g) at (\x,-1+\y) {};
\node[nod] (h) at (.707+\x,-.707+\y) {};
\draw[line width=1pt] (g) -- (f) -- (e); 
\draw[line width=1pt] (c) -- (d) -- (f);
\draw[line width=1pt] (h) -- (g) -- (a); \draw[line width=1pt] (d) -- (b); 
\node (o) at (\x,\y)   {18};
\def\x{6}
\def\y{3}
\node[nod] (a) at (1+\x,\y) {};
\node[nod] (b) at (.707+\x,.707+\y) {};
\node[nod] (c) at (0+\x,1+\y) {};
\node[nod] (d) at (-.707+\x,.707+\y) {};
\node[nod] (e) at (-1+\x,0+\y) {};
\node[nod] (f) at (-.707+\x,-.707+\y) {};
\node[nod] (g) at (\x,-1+\y) {};
\node[nod] (h) at (.707+\x,-.707+\y) {};
\draw[line width=1pt] (g) -- (f) -- (e) -- (d); 
\draw[line width=1pt] (b) -- (c) -- (f);
\draw[line width=1pt] (h) -- (g) -- (a);
\node (o) at (\x+.2,\y)   {48};
\def\x{9}
\def\y{3}
\node[nod] (a) at (1+\x,\y) {};
\node[nod] (b) at (.707+\x,.707+\y) {};
\node[nod] (c) at (0+\x,1+\y) {};
\node[nod] (d) at (-.707+\x,.707+\y) {};
\node[nod] (e) at (-1+\x,0+\y) {};
\node[nod] (f) at (-.707+\x,-.707+\y) {};
\node[nod] (g) at (\x,-1+\y) {};
\node[nod] (h) at (.707+\x,-.707+\y) {};
\draw[line width=1pt] (g) -- (f) -- (e) -- (d); 
\draw[line width=1pt] (c) -- (f);
\draw[line width=1pt] (h) -- (g) -- (a) -- (b);
\node (o) at (\x+.2,\y) {60};
\def\x{12}
\def\y{3}
\node[nod] (a) at (1+\x,\y) {};
\node[nod] (b) at (.707+\x,.707+\y) {};
\node[nod] (c) at (0+\x,1+\y) {};
\node[nod] (d) at (-.707+\x,.707+\y) {};
\node[nod] (e) at (-1+\x,0+\y) {};
\node[nod] (f) at (-.707+\x,-.707+\y) {};
\node[nod] (g) at (\x,-1+\y) {};
\node[nod] (h) at (.707+\x,-.707+\y) {};
\draw[line width=1pt] (g) -- (f) -- (e); 
\draw[line width=1pt] (d) -- (f);
\draw[line width=1pt] (h) -- (g) -- (a); \draw[line width=1pt] (b) -- (c) -- (d); 
\node (o) at (\x,\y)   {79};
\def\x{15}
\def\y{3}
\node[nod] (a) at (1+\x,\y) {};
\node[nod] (b) at (.707+\x,.707+\y) {};
\node[nod] (c) at (0+\x,1+\y) {};
\node[nod] (d) at (-.707+\x,.707+\y) {};
\node[nod] (e) at (-1+\x,0+\y) {};
\node[nod] (f) at (-.707+\x,-.707+\y) {};
\node[nod] (g) at (\x,-1+\y) {};
\node[nod] (h) at (.707+\x,-.707+\y) {};
\draw[line width=1pt] (g) -- (f) -- (e) -- (d); 
\draw[line width=1pt] (b) -- (c) -- (f);
\draw[line width=1pt] (a) -- (h) -- (g);
\node (o) at (\x+.2,\y)  {126};
\def\x{0}
\def\y{0}
\node[nod] (a) at (1+\x,\y) {};
\node[nod] (b) at (.707+\x,.707+\y) {};
\node[nod] (c) at (0+\x,1+\y) {};
\node[nod] (d) at (-.707+\x,.707+\y) {};
\node[nod] (e) at (-1+\x,0+\y) {};
\node[nod] (f) at (-.707+\x,-.707+\y) {};
\node[nod] (g) at (\x,-1+\y) {};
\node[nod] (h) at (.707+\x,-.707+\y) {};
\draw[line width=1pt] (a) -- (h) -- (g) -- (f) -- (e); 
\draw[line width=1pt] (b) -- (h);
\draw[line width=1pt] (f) -- (d) -- (c); 
\node (o) at (\x,\y)   {79};
\def\x{3}
\def\y{0}
\node[nod] (a) at (1+\x,\y) {};
\node[nod] (b) at (.707+\x,.707+\y) {};
\node[nod] (c) at (0+\x,1+\y) {};
\node[nod] (d) at (-.707+\x,.707+\y) {};
\node[nod] (e) at (-1+\x,0+\y) {};
\node[nod] (f) at (-.707+\x,-.707+\y) {};
\node[nod] (g) at (\x,-1+\y) {};
\node[nod] (h) at (.707+\x,-.707+\y) {};
\draw[line width=1pt] (a) -- (h) -- (g) -- (f) -- (e); 
\draw[line width=1pt] (f) -- (d) -- (c) -- (b); 
\node (o) at (\x,\y)   {112};
\def\x{6}
\def\y{0}
\node[nod] (a) at (1+\x,\y) {};
\node[nod] (b) at (.707+\x,.707+\y) {};
\node[nod] (c) at (0+\x,1+\y) {};
\node[nod] (d) at (-.707+\x,.707+\y) {};
\node[nod] (e) at (-1+\x,0+\y) {};
\node[nod] (f) at (-.707+\x,-.707+\y) {};
\node[nod] (g) at (\x,-1+\y) {};
\node[nod] (h) at (.707+\x,-.707+\y) {};
\draw[line width=1pt] (b) -- (a) -- (h) -- (g) -- (f) -- (e); 
\draw[line width=1pt] (f) -- (d) -- (c);
\node (o) at (\x,\y) {207};
\def\x{9}
\def\y{0}
\node[nod] (a) at (1+\x,\y) {};
\node[nod] (b) at (.707+\x,.707+\y) {};
\node[nod] (c) at (0+\x,1+\y) {};
\node[nod] (d) at (-.707+\x,.707+\y) {};
\node[nod] (e) at (-1+\x,0+\y) {};
\node[nod] (f) at (-.707+\x,-.707+\y) {};
\node[nod] (g) at (\x,-1+\y) {};
\node[nod] (h) at (.707+\x,-.707+\y) {};
\draw[line width=1pt] (b) -- (a) -- (h) -- (g) -- (f) -- (e); 
\draw[line width=1pt] (f) -- (d);
\draw[line width=1pt] (a) -- (c);
\node (o) at (\x,\y) {30};
\def\x{12}
\def\y{0}
\node[nod] (a) at (1+\x,\y) {};
\node[nod] (b) at (.707+\x,.707+\y) {};
\node[nod] (c) at (0+\x,1+\y) {};
\node[nod] (d) at (-.707+\x,.707+\y) {};
\node[nod] (e) at (-1+\x,0+\y) {};
\node[nod] (f) at (-.707+\x,-.707+\y) {};
\node[nod] (g) at (\x,-1+\y) {};
\node[nod] (h) at (.707+\x,-.707+\y) {};
\draw[line width=1pt] (c) -- (b) -- (a) -- (h) -- (g) -- (f) -- (e); 
\draw[line width=1pt] (f) -- (d);
\node (o) at (\x,\y) {144};
\def\x{15}
\def\y{0}
\node[nod] (a) at (1+\x,\y) {};
\node[nod] (b) at (.707+\x,.707+\y) {};
\node[nod] (c) at (0+\x,1+\y) {};
\node[nod] (d) at (-.707+\x,.707+\y) {};
\node[nod] (e) at (-1+\x,0+\y) {};
\node[nod] (f) at (-.707+\x,-.707+\y) {};
\node[nod] (g) at (\x,-1+\y) {};
\node[nod] (h) at (.707+\x,-.707+\y) {};
\draw[line width=1pt] (c) -- (b) -- (a) -- (h) -- (g) -- (f) -- (e) -- (d); 
\node (o) at (\x,\y) {195};
\end{tikzpicture}
\caption{\label{t8} Trees of order 8 and the number of LV-equivalent nonisomorphic hypergraphs.}
\end{figure}
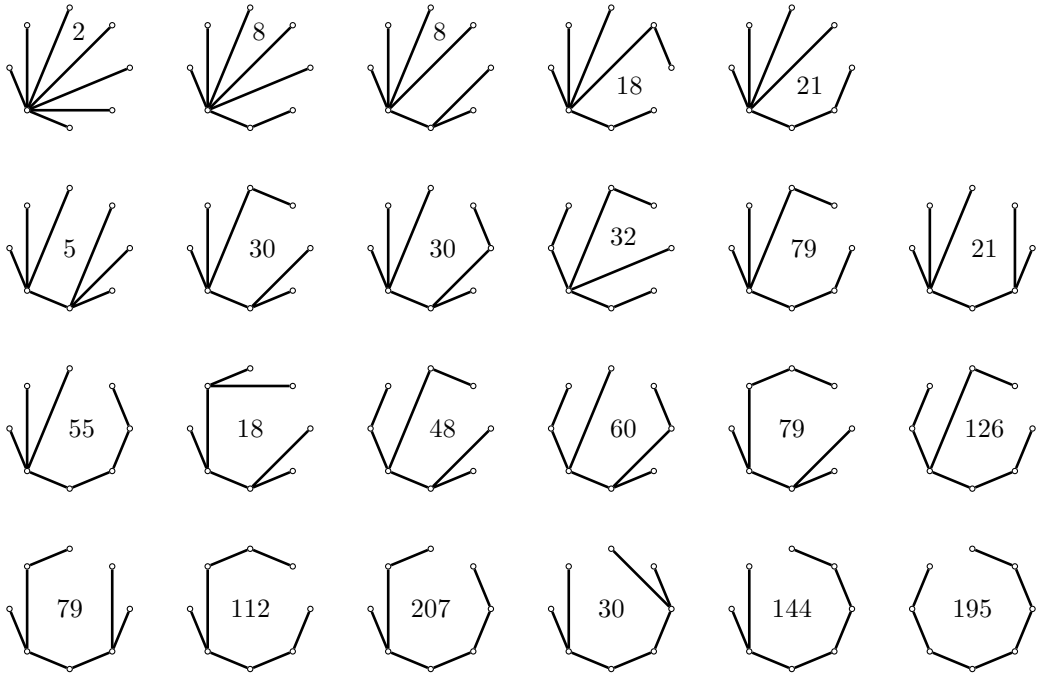

\noindent
{\bf Acknowledgement.} The author would like to thank Reinout Quispel, who found the condition for \eqref{DP} to be a 3-DP of the Lotka-Volterra system \eqref{LVsys}. He is also grateful to an anonymous referee who pointed out that the previous version of Lemma \ref{lem} did not seem to be powerful enough, which led to the subsequent propositions.

\end{document}